\documentclass[10pt,twocolumn,notitlepage,aps,pra,showpacs,superscriptaddress,longbibliography,nofootinbib]{revtex4-2} 

\usepackage[T1]{fontenc}

\usepackage[utf8]{inputenc} 
\usepackage{amssymb,mathtools}
\usepackage{graphicx}
\usepackage{amsthm}
\usepackage{amssymb, latexsym, bm, mathtools, braket, multirow,enumitem}

\usepackage[cmintegrals]{newtxmath}
\usepackage{color}

\DeclareFontFamily{OMX}{MnSymbolE}{}
\DeclareFontShape{OMX}{MnSymbolE}{m}{n}{
    <-6>  MnSymbolE5
   <6-7>  MnSymbolE6
   <7-8>  MnSymbolE7
   <8-9>  MnSymbolE8
   <9-10> MnSymbolE9
  <10-12> MnSymbolE10
  <12->   MnSymbolE12}{}
\DeclareSymbolFont{mnlargesymbols}{OMX}{MnSymbolE}{m}{n}
\SetSymbolFont{mnlargesymbols}{bold}{OMX}{MnSymbolE}{b}{n}
\DeclareMathDelimiter{\llangle}{\mathopen}{mnlargesymbols}{'164}{mnlargesymbols}{'164}
\DeclareMathDelimiter{\rrangle}{\mathclose}{mnlargesymbols}{'171}{mnlargesymbols}{'171}

\theoremstyle{plain}
\newtheorem{lemm}{Lemma}

\newtheorem{theo}{Theorem}
\newtheorem{prop}{Proposition}

\newtheorem{coro}{Corollary}

\theoremstyle{definition}
\newtheorem{remark}{Remark}
\newtheorem{prot}{Protocol}

\newcommand{\cD}{\mathcal{D}}
\newcommand{\cE}{\mathcal{E}}

\newcommand*{\QEDA}{\hfill\ensuremath{\blacksquare}}%
\newcommand{\ii}{\iota}

\DeclareMathOperator{\Tr}{Tr}

\DeclareMathOperator{\id}{id}

\newcommand\sff{\mathsf{f}}
\newcommand\sg{\mathsf{g}}
\newcommand\sh{\mathsf{h}}

\newcommand\sm{\mathsf{m}}

\newcommand\sr{\mathsf{r}}

\newcommand\CC{\mathrm{CC}(\Phi)}

\newcommand{\reR}{R}
\newcommand{\reQ}{Q}
\newcommand{\reAns}{\mathop{Ans}}
\newcommand{\fans}{\mathsf{ans}}

\newcommand\sA{\mathsf{A}}
\newcommand\sB{\mathsf{B}}
\newcommand\sC{\mathsf{C}}

\newcommand\sI{\mathsf{I}}
\newcommand\sM{\mathsf{M}}

\newcommand\sX{\mathsf{X}}

\newcommand\sZ{\mathsf{Z}}

\newcommand\bZ{\mathbb{Z}}

\usepackage{tikz}
\usetikzlibrary{arrows}
\usetikzlibrary{shapes,snakes, quantikz}
\usetikzlibrary{tikzmark}
\usetikzlibrary{fit,backgrounds}
\usetikzlibrary{matrix,decorations.pathreplacing,angles,quotes}
\usetikzlibrary{patterns}
\usetikzlibrary{calc}
\usetikzlibrary{positioning}
\tikzstyle{block} = [rectangle, draw, 
    minimum width=3em, text centered, minimum height=12em]
\tikzstyle{rblock} = [rectangle, draw, 
    minimum width=5em, text centered, minimum height=4em, rounded corners]   
\tikzstyle{circ} = [circle, draw, inner sep=0em, minimum width=2em,
     text centered]
\tikzstyle{circ2} = [circle, draw, inner sep=0em, minimum width=2.2em,
     text centered]    
\tikzstyle{rec} = [rectangle, draw]
\tikzstyle{line} = [draw, -latex]
\tikzstyle{imp-line} = [draw, -implies, double equal sign distance]

\tikzset{snake arrow/.style=
{
decorate,
decoration={snake,amplitude=.4mm,segment length=2mm,pre length=1mm, post length=1mm}}
}
\tikzset{snake arrow0/.style=
{
decorate,
decoration={snake,amplitude=.4mm,segment length=2mm,post length=1mm}}
}

\def\Label#1{\label{#1}\ [\ \text{#1}\ ]\ }
\def\Label{\label}
\newcommand{\red}[1]{{\leavevmode\color{red}#1}}

\begin{document}
\title{Prior Entanglement Exponentially Improves One-Server 
Quantum Private Information Retrieval for Quantum Messages}
\author{Seunghoan~Song}
\email{seunghoans@gmail.com}
\affiliation{Graduate School of Mathematics, Nagoya University, Chikusa-ku, Nagoya, 464-8602, Japan}
\author{François Le Gall}
\email{legall@math.nagoya-u.ac.jp}
\affiliation{Graduate School of Mathematics, Nagoya University, Chikusa-ku, Nagoya, 464-8602, Japan}
\author{Masahito Hayashi}
\email{hmasahito@cuhk.edu.cn, hayashi@iqasz.cn}
\affiliation{School of Data Science, The Chinese University of Hong Kong, Shenzhen, Longgang District, Shenzhen, 518172, China}
\affiliation{International Quantum Academy (SIQA), Futian District, Shenzhen 518048, China}
\affiliation{Graduate School of Mathematics, Nagoya University, Chikusa-ku, Nagoya, 464-8602, Japan}

\begin{abstract}
Quantum private information retrieval (QPIR) for quantum messages is a quantum communication task, in which a user retrieves one of the multiple quantum states from the server without revealing which state is retrieved.
In the one-server setting, we find an exponential gap in the communication complexities between the presence and absence of prior entanglement in this problem with the one-server setting.
To achieve this aim, as the first step, we prove that the trivial solution of downloading all messages 
is optimal under QPIR for quantum messages, which is a similar result to that of classical PIR but different from QPIR for classical messages.
As the second step, we propose an efficient one-server
one-round QPIR protocol with prior entanglement by constructing a reduction from a QPIR protocol for classical messages to a QPIR protocol for quantum messages in the presence of prior entanglement.
\end{abstract}

\maketitle

\section{Introduction}
\subsection{Private information retrieval (PIR)}
Entanglement is a valuable resource for quantum information processing, 
enabling various tasks including quantum teleportation \cite{BBCJPW93} 
and dense coding, also known as entanglement-assisted communication \cite{BW}. 
Although entanglement-assisted communication enhances the speed 
not only for conventional communication but also for secret communication,
their improvements are limited to constant times \cite{WDLLL,WLH}.
For further development of entanglement-assisted communication,
we need to find significant improvement by entanglement-assisted communication.

For this aim, we focus on private information retrieval (PIR),
a task in which a user retrieves a message from a server without revealing which message has been retrieved, when the server possesses multiple messages. 
Many papers \cite{KdW03, KdW04,Ole11,BB15,LeG12,KLLGR16,ABCGLS19,SH19,SH19-2,SH20,AHPH20,ASHPHH21,KL20,WKNL21-1,WKNL21-2}
studied Quantum PIR (QPIR), i.e., PIR using quantum states, 
when the intended messages are given as the classical messages.
This problem setting is simplified to C-QPIR. 
On the other hand, since various types of quantum information processings require the transmission of quantum states, i.e., 
the quantum messages \cite{Wie83,GC01,Moc07,CK09,ACG+16},
it is needed to develop QPIR for quantum messages, which is simplified to Q-QPIR, 
while no preceding paper studied this topic.

In this paper, to enhance quantum information technology,
we study private information retrieval for quantum messages with one server,
and present an exponential speedup through the use of prior entanglement
as a significant improvement.
Although there have been mainly two approaches: PIR with computational assumptions \cite{CMS99, Lipmaa10} and PIR with multiple servers \cite{BS03,Yekanin07,DGH12},
recent attention has focused on information-theoretic aspects of PIR \cite
{CHY15,SJ17, SJ17-2, SJ18, BU18, FHGHK17, KLRG17, LKRG18, TSC18-2, WS17,  Tandon17,  BU19,L19, KGHERS19, TGKFH19}.
In this paper, we solely consider one-server QPIR without computational assumptions.

\subsection{QPIR for classical messages}
PIR has also been studied when quantum communication is allowed between the user and the server  \cite{KdW03,KdW04,Ole11,BB15,LeG12,KLLGR16,ABCGLS19,SH19,SH19-2,SH20,AHPH20,ASHPHH21,KL20,WKNL21-1,WKNL21-2}.
These papers consider the case when the total number of bits in the messages is $\sm$.
For the secrecy in C-QPIR, 
we often focus on the potential information leakage in all rounds, which is called
the {\it all-round criterion} in this paper
and has been studied under several security models.
One is the {\em honest-server model}, 
in which,
we discuss the user's secrecy only when the server is honest, i.e., 
the server does not deviate from the protocol.
The other is the {\em specious-server model}, in which,
we discuss the user's secrecy even when the server deviates from the protocol as far as its dishonest operations are not revealed to the user, which is called {\em specious adversary}.
The secrecy under the specious-server model has a stronger requirement than 
the secrecy under the honest-server model.
Interestingly, under the honest-server model,
Le Gall \cite{LeG12} proposed a C-QPIR protocol with communication complexity 
$O(\sqrt{\sm})$ in the all-round criterion, and
Kerenidis et al. \cite{KLLGR16} improved this result to $O(\mathrm{poly}\log \sm)$
in another criterion, 
where the communication complexity in the quantum case is the total number of communicated qubits.
However, when 
the specious-server model is adopted and the possible input states are extended to arbitrary superposition states,
Baumeler and Broadbent \cite{BB15} proved that the communication complexity is at least $\Theta(\sm)$, i.e., the trivial solution of downloading all messages is optimal also for this case.
Even when prior entanglement is allowed between the user and the server,
the communication complexity is also lower bounded by $\Theta(\sm)$ under the specious-server model with the above extended possible input states
\cite{ABCGLS19}.
Therefore, 
the advantage of prior entanglement is limited under the specious-server model with the above extended possible input states.
In contrast, 
prior entanglement might potentially have polynomial improvement under the honest-server model, 
but it is still unclear how much prior entanglement improves communication complexity under the honest-server model.

When the server truly follows the protocol, the information obtained by the server
is limited to the server's final state.
Hence, the information leakage in the server's final state can be considered as another criterion, which is called the {\it final-state criterion}.
While the final-state criterion under the honest-server model is a too weak setting,
it is reasonable to consider the final-state criterion under the specious-server model, which is essentially 
equivalent to the cheat-sensitive setting studied in \cite{GLM}.

\begin{figure}[t]
\begin{center}
        \resizebox {1\linewidth} {!} {
\begin{tikzpicture}[scale=0.5, node distance = 3.3cm, every text node part/.style={align=center}, auto]
	\node [rblock] (serv) {Server};
	\node [rblock, right=12em of serv] (user) {User};
	\node [above=3em of serv] (1) {1. message states\\$\rho_1, \ldots , \rho_{\sff}$};
	\node [above=3em of user] (2) {2. target index $k\in[\sff]$};
	\node [below=3em of user] (5) {5. retrieved state $\rho_k$};

	\path [line] (serv.22 -| user.west) --node[above]{3. queries\\$Q^{(1)},\ldots,Q^{(\sr)} $} (serv.22);
	\path [line] (serv.-22) --node[below]{4. answers\\$A^{(1)}, \ldots,A^{(\sr)}$} (serv.-22 -| user.west);
	\path [line] (1) -- (serv.north);
	\path [line] (2) -- (user.north);
	\path [line] (user.south) -| (5);
	
\end{tikzpicture}
}
\caption{One-server QPIR protocol with quantum messages. At round $i$, the user uploads a query $Q^{(i)}$ and downloads an answer $A^{(i)}$.}
\Label{fig:one-server}
\end{center}
\end{figure}
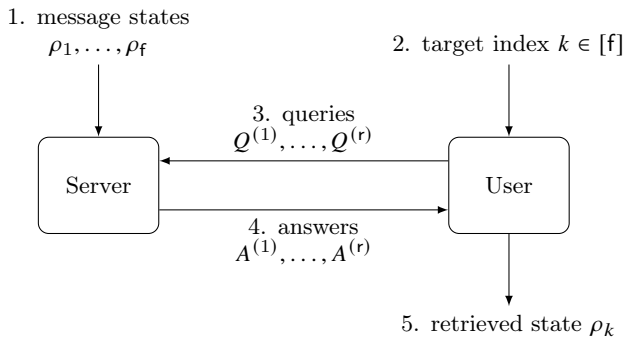

\subsection{Our contributions} 
In this paper, for Q-QPIR protocols and the total number $\sm$ of qubits,
we show that the communication complexity is at least 
$\Theta(\sm)$, i.e., the trivial solution of downloading all messages is optimal for one-server Q-QPIR even in the final-state criterion and even with the honest-server model
if prior entanglement is not allowed between the server and the user.
This fact shows that prior entanglement between the server and the user
is necessary for further improvement 
under the one-server model 
even for Q-QPIR under the honest-server model, the weakest secrecy requirement.
To overcome this problem, we propose a one-server Q-QPIR protocol with 
prior entanglement between the server and the user,
which achieves the communication complexity $O(\log \sm)$.
That is, prior entanglement has exponential improvement for 
Q-QPIR under the honest-server model.


\subsection{Organization of this paper}
The remainder of the paper is organized as follows.
Section \ref{S2} gives the definitions of several concepts
and the outline of our results including the comparison with existing results.
Section~\ref{sec:prelim} is the technical preliminaries of the paper.
Section~\ref{C-QPIR-pro} presents our results for 
C-QPIR protocol with communication complexity $O( \log \sm)$.
Section~\ref{sec:blind-one-opt}
derives the lower bound of the communication complexity for Q-QPIR in the final-state criterion under the honest-server model when prior entanglement is not shared.
Section~\ref{sec:ent} proposes an efficient Q-QPIR protocol with prior entanglement under various settings.
Section~\ref{sec:conclusion} is the conclusion of the paper.

\section{Definitions and Outline of our Results}\Label{S2}
\subsection{Definitions of various concepts}
To briefly explain our results, we prepare 
the definitions of various concepts to cover C-QPIR protocols and
Q-QPIR protocols in a common framework. 
\subsubsection{Correctness, complexity, and unitary-type}
To discuss the properties of our QPIR protocols, 
we prepare several concepts.
First, we define the set $\mathcal{S}$ of possible quantum states 
as a subset of the set ${\cal S}({\cal H}_d)$ of states on $\mathbb{C}^d $.
A QPIR protocol is called a QPIR protocol with $\mathbb{C}^d $
over the set $\mathcal{S}$ 
when it works when the set $\mathcal{S}$ is the set of possible quantum states. 
For example, 
when $\mathcal{S}$ is the set ${\cal C}$ of orthogonal pure states $\{ |j\rangle\}_{j=0}^{d-1}$,
a QPIR protocol is a C-QPIR protocol discussed in \cite{BB15}.
In contrast, when $\mathcal{S}$ is 
the set ${\cal Q}$ of all pure states on the system $\mathbb{C}^d $,
a QPIR protocol is a Q-QPIR protocol.
When we do not identify the set $\mathcal{S}$, we consider that it is given as the above case.
We denote the number of messages by $\sff $. 
A QPIR protocol $\Phi$ has two types of inputs.
The first input is composed of $\sff$ messages, whose systems are written as
${\cal H}_1, \ldots, {\cal H}_{\sff}$.
Their state is written as $\sff$ states $(\rho_1, \ldots, \rho_\sff) \in 
\mathcal{S}^\sff$.
The second input is the choice of the label of the message intended by the user, which is written as the random variable $K$.
The quantum system to describe the variable $K$ is denoted by ${\cal K}$.
We denote the remaining initial user's and server's systems 
by ${\cal R}_u$ and ${\cal R}_s$, respectively.
The output of the protocol is a state $\rho_{out}$ on ${\cal H}_d$.

A QPIR protocol $\Phi$ has bilateral communication.
The communication from the user to the servers is the upload communication,
and  
the communication from the servers to the users is the download communication.
The communication complexity is composed of 
the upload complexity and the download complexity.
The upload complexity is the sum of the communication sizes of all upload communications, and
the download complexity is the sum of the communication sizes of all 
download communications.
The sum of the upload and download complexity is called
the communication complexity.
We adopt the communication complexity
as the optimality criterion under various security conditions.

A QPIR protocol $\Phi$ is called a deterministic protocol 
when the following two conditions hold. 
The upload complexity and the download complexity are determined only by the protocol $\Phi$.
When the user and the servers are honest,
the output is determined only by $(\rho_1, \ldots, \rho_\sff)$ and $K$.
When $\Phi$ is a deterministic protocol, we denote the output state 
by $\Phi_{out}(\rho_1, \ldots, \rho_\sff,K)= \rho_{out}$.
The upload complexity, the download complexity, and the communication complexity are 
denoted by $UC(\Phi)$, $DC(\Phi)$, and $CC(\Phi)$, respectively.
Hence, the communication complexity $CC(\Phi)$ is calculated as
$UC(\Phi)+DC(\Phi)$.
A protocol $\Phi$ is called correct when 
the protocol is a deterministic protocol and 
the relation $\Phi_{out}(\rho_1, \ldots, \rho_\sff,k)=\rho_k$ holds
for any elements $k \in [\sff]$ and 
$(\rho_1, \ldots, \rho_\sff) \in \mathcal{S}^\sff$.

Another important class of QPIR protocols is unitary-type protocols.
When a QPIR protocol $\Phi$ satisfies the following conditions,
it is called {\it unitary-type}.
\begin{itemize}
\item
The initial states $\rho_{{\cal R}_s}$ on ${\cal R}_s$
and $\rho_{{\cal R}_u}$ on ${\cal R}_u$ are pure.
\item
At each round, both the user and the server apply only unitary operations to 
the systems under their control.
\item
A measurement is done only when the user reads out the message as the outcome of the protocol.
\end{itemize}
The reference \cite{ABCGLS19} refers to the above property as measurement-free due to the third condition while it assumes the first and second conditions implicitly.
Since the first and second conditions are more essential, we call it unitary-type.

\subsubsection{Secrecy}
In this paper, we address only the secrecy of the user's choice.
There are two security criteria.
One is the final-state criterion, in which,
it is required that the server's final state does not depend on the user's 
choice $K$.
The other is the all-round criterion, in which,
it is required that the server's state in any round does not depend on the user's 
choice $K$.
When we consider the secrecy, we may extend the set of possible inputs to 
$\tilde{\cal S}$ that includes the set ${\cal S}$.
For example, 
in the case of C-QPIR, the set ${\cal S}$ is given as
the set ${\cal C}$. Then, 
we can choose $\tilde{\cal S}$ as the set ${\cal C}$ or ${\cal Q}$.
The case with $\tilde{\cal S}={\cal C}$ is called the classical input case,
and the case with $\tilde{\cal S}={\cal Q}$ is called the superposition input case.
Instead, in the case of Q-QPIR, the set ${\cal S}$ is given as
the set ${\cal Q}$. Hence, the set $\tilde{\cal S}$ is chosen as the same set ${\cal Q}$.

Even when we fix the security criterion and the sets ${\cal S}$ and $\tilde{\cal S}$,
there still exist three models for the secrecy for a QPIR protocol $\Phi$.
The first one is the honest-server model, which assumes that 
the servers are honest.
We say that a QPIR protocol $\Phi$ satisfies the secrecy in the final-state criterion
under the honest-server model with input states $\tilde{\cal S}$
when the following condition holds. 
When the user and the servers are honest,
the server has no information for $K$ in the final state, i.e., the relation 
\begin{align}
\rho_{S,F}(\rho_1, \ldots, \rho_\sff,k)=\rho_{S,F}(\rho_1, \ldots, \rho_\sff,k')
\Label{MLK}
\end{align}
holds for any $k,k' \in [\sff]$ and
$(\rho_1, \ldots, \rho_\sff) \in \tilde{\mathcal{S}}^\sff$,
where $\rho_{S,F} (\rho_1, \ldots, \rho_\sff,K)$ 
is the final state on the server dependent of the variable $K$.
In the condition \eqref{MLK}, the states $\rho_k$ is chosen from 
$\tilde{\mathcal{S}}$, not from ${\mathcal{S}}$.
We say that a QPIR protocol $\Phi$ satisfies the secrecy in the all-round criterion
under the honest-server model with input states $\tilde{\cal S}$
when the following condition holds,
the server has no information for $K$ in all rounds, i.e.,
the relation 
\begin{align}
\rho_{S,j}(\rho_1, \ldots, \rho_\sff,k)=\rho_{S,j}(\rho_1, \ldots, \rho_\sff,k')
\Label{MLK2}
\end{align}
holds for any $k,k' \in [\sff]$ and
$(\rho_1, \ldots, \rho_\sff) \in \tilde{\mathcal{S}}^\sff$,
where $\rho_{S,j} (\rho_1, \ldots, \rho_\sff,K)$ 
is the state on the server dependent of
the variable $K$ when the server receives the query in the $j$-th round.
The following is the meaning of the secrecy 
in the all-round criterion
under the honest-server model.
Assume that the user and the server are honest.
Even when the server stops the protocol at the $j$-th round for any $j$,
the server cannot obtain any information for $K$.

The second model is called the specious-server model introduced in \cite{DNS10}.
When the server applies other operations that deviate from the original protocol,
such an operation is called an attack. 
An attack of the server is called a specious attack
when the attack satisfies the following conditions.
The server sends the answer at the time specified by the protocol, but the contents of 
the answer do not follow the protocol. 
Also, the server does not access the information under the control of the user.
In addition, the attack is not revealed to the user
under the condition that the user is honest, i.e.,
there exists the server's operation ${\cal F}_{S,j}$ such that
the relation 
\begin{align}
({\cal F}_{S,j}\otimes \iota) 
\tilde{\rho}_{j}(\rho_1, \ldots, \rho_\sff,k)=
\rho_{j}(\rho_1, \ldots, \rho_\sff,k)
\Label{MLK3}
\end{align}
holds for any $k,k' \in [\sff]$ and
$(\rho_1, \ldots, \rho_\sff) \in \tilde{\mathcal{S}}^\sff$,
where 
$\rho_{j}(\rho_1, \ldots, \rho_\sff,K)$ 
($\tilde{\rho}_{j}(\rho_1, \ldots, \rho_\sff,K)$) 
is the state on the whole system dependently of
the variable $K$ when the user receives the answer in the $j$-th round
under the assumption that the user is honest and 
the server is honest (the server makes the attack).
Notice that the definition of a specious attack depends on 
the choice of the set $\tilde{\cal S}$.
The meaning of \eqref{MLK3} is the following.
When the user decides to stop the protocol to check whether 
the server follows the protocol
after the user receives the answer in the $j$-th round,
the user asks the server to submit the evidence that 
the server follows the protocol.
Then, the server sends his system after applying the operation ${\cal F}_{S,j}$.
When $\tilde{\cal S} $ is chosen to be the set ${\cal Q}$ of pure states,
a specious attack coincides with a 0-specious adversary
in the sense of \cite[Definition 2.4]{ABCGLS19}
because it is sufficient to check the case with even $t$ in \cite[Definition 2.4]{ABCGLS19}.
Also, when $\tilde{\cal S} $ is chosen to be the set ${\cal C}$,
the secrecy in the all-round criterion under the specious server model 
coincides with the anchored 0-privacy under 0-specious servers \cite{ABCGLS19}.

We say that a QPIR protocol $\Phi$ satisfies the secrecy in the final-state criterion 
(the all-round criterion)
under the specious-server model with input states $\tilde{\cal S}$ 
when the following condition holds. 
When a server performs a specious attack and the user is honest,
the server obtains no information about the user's request $K$
in all rounds, i.e., the condition \eqref{MLK} (the condition \eqref{MLK2}) holds.
In fact, the secrecy condition in the final-state criterion is weaker than 
the secrecy condition in the all-round criterion even under the specious-server model.
The secrecy condition in the final-state criterion under the specious-server model
is essentially equivalent to the cheat-sensitive secrecy condition considered in \cite{GLM}.

The third model is called the dishonest-server model.
We say that a QPIR protocol $\Phi$ satisfies the secrecy under 
the dishonest-server model
when the following condition holds. 
When the server applies an attack and
the user is honest,
the server obtains no information of the user's request $K$, i.e., 
the condition \eqref{MLK} holds.
In the dishonest-server model, the server is allowed to make any attack
under the following condition.
The server sends the answer at the time specified by the protocol, but the contents of the answer do not follow the protocol. 
Also, the server does not access the information under the control of the user.
Thus, the server can send any information on each round 
under this condition.
Hence, the ability of the attack does not depend on the set $\tilde{\cal S}$.
Also, the server can store the state received in any round.
Hence, 
the server can obtain the same information in the final state as the information in the $j$-th round.

Further, when the protocol has only one round
and we adopt the all-round criterion,
there is no difference among 
the honest-server model, the specious-server model, and
the dishonest-server model 
because all information obtained by the server is reduced to the state on the server 
when the server received the query in the first round. 
As a result, the information obtained by the server does not depend on 
the server's operation, i.e., the server's attack model.

\begin{remark}
In the papers \cite{BB15,ABCGLS19},
the security against specious adversaries means
the secrecy in the all-round criterion under the specious-server model with input states
${\cal Q}$ for C-QPIR in our definition.
Instead, in the paper \cite{ABCGLS19},
the anchored specious security means
the secrecy in the all-round criterion under the specious-server model with input states
${\cal C}$ for C-QPIR in our definition.
The papers \cite{BB15,ABCGLS19} did not consider the final-state criterion.
\end{remark}

\begin{table*}[ht]
\caption{Optimal communication complexity of one-server C-QPIR}\Label{tab:existing-results}
\begin{center}
\begin{tabular}{|c|c|c|c|c|c|c|}
\hline
\multirow{2}{*}{security} && & \multicolumn{4}{c|}{Optimal communication complexity}
\\
 \cline{4-7}
 \multirow{2}{*}{criterion}
&  input & server & \multicolumn{2}{c|}{without PE} & \multicolumn{2}{c|}{with PE}
\\
 \cline{4-7}
&&& one-round & multi-round& one-round & multi-round \\
\hline
\multirow{4}{*}{final-} & \multirow{4}{*}{classical} & \multirow{2}{*}{honest} 
&$ O(\log \sm)$*  &$ O(\log \sm)$* &
$ O(\log \sm)$* &$ O(\log \sm)$* \\
& & &[Section \ref{sec-one}] &[Section \ref{sec-one}]&
[Section \ref{sec-one}]& \cite{KLLGR16}+[Lemma \ref{prop:sec2}] \\
 \cline{3-7}
 && \multirow{2}{*}{specious} & $ O(\log \sm)$* 
&$ O(\log \sm)$* &$ O(\log \sm)$* 
& $O(\log \sm)$*  \\
 && &  [Section \ref{sec-one2}]
& [Section \ref{sec-one2}]& [Section \ref{sec-one2}]
& \cite{KLLGR16}+[Corollary \ref{CCor1}] \\
 \cline{2-7}
\multirow{3}{*}{state} & \multirow{3}{*}{superposition} & \multirow{2}{*}{honest} &
$ O(\log \sm)$* & $ O(\log \sm)$* &
$ O(\log \sm)$* & $ O(\log \sm)$*  \\
& &  &
  [Section \ref{sec-one}]& [Section \ref{sec-one}]&
  [Section \ref{sec-one}]& [Section \ref{sec-one}] \\
 \cline{3-7}
 & & specious & ?&?&?&? \\
\hline
\multirow{3}{*}{all-} & \multirow{4}{*}{classical} & \multirow{2}{*}{honest} & &$O(\mathrm{poly}\log \sm)$* && $O(\log \sm)$* \\
& & & & \cite{KLLGR16}+[Lemma \ref{prop:sec1}]&& 
\cite{KLLGR16}+[Lemma \ref{prop:sec2}] \\
 \cline{3-3}\cline{5-5} \cline{7-7}
 & & \multirow{2}{*}{specious} & &$O(\mathrm{poly}\log \sm)$*&& $O(\log \sm)$* \\
\multirow{3}{*}{round}
 & & &$ \Theta(\sm)$ \cite{BB15} & \cite{KLLGR16}+[Corollary \ref{CCor1}]&
 $ \Theta(\sm)$ \cite{ABCGLS19}& \cite{KLLGR16}+[Corollary \ref{CCor1}]\\
  \cline{2-3}\cline{5-5} \cline{7-7}
 & \multirow{2}{*}{superposition} & honest &
 & ? && ? \\
 \cline{3-3}\cline{5-5} \cline{7-7}
& & specious & & $ \Theta(\sm)$ \cite{BB15}&& $ \Theta(\sm)$ \cite{ABCGLS19}\\
 \cline{1-3} \cline{5-5} \cline{7-7}
\multicolumn{3}{|c|}{dishonest} & &$ \Theta(\sm)$ \cite{BB15}&& $ \Theta(\sm)$ \cite{ABCGLS19}\\
\hline
\end{tabular}
\end{center}
\begin{flushleft}
$\sm$ is the total size of messages.
* expresses the case when each message size is fixed.
The symbol \cite{KLLGR16}+ [Lemma \ref{prop:sec2}] shows that 
the protocol was proposed in \cite{KLLGR16}, but
its secrecy is shown in Lemma \ref{prop:sec2} of this paper. 
These notations are applied to Table \ref{tab:blind} as well. 
\end{flushleft}
\end{table*}

\begin{table*}[ht]
\caption{Optimal communication complexity of one-server Q-QPIR}\Label{tab:blind}
\begin{center}
\begin{tabular}{|c|c|c|c|c|c|}
\hline
\multirow{2}{*}{security} &&  \multicolumn{4}{c|}{Optimal communication complexity}
\\
 \cline{3-6}
 \multirow{2}{*}{criterion}
& server  & \multicolumn{2}{c|}{without PE} & \multicolumn{2}{c|}{with PE}
\\
 \cline{3-6}
&& one-round & multi-round& one-round & multi-round \\
\hline
final- & honest &
$ \Theta(\sm)$ [Theorem \ref{theo:multiround}] &
$ \Theta(\sm)$ [Theorem \ref{theo:multiround}]&
$ O(\log \sm)$* [Corollary \ref{CO2}]&$ O(\log \sm)$* [Corollary \ref{CO2}] \\
 \cline{2-6}
state & specious & $ \Theta(\sm)$ [Theorem \ref{theo:multiround}] &
$ \Theta(\sm)$ [Theorem \ref{theo:multiround}] &
$ O(\log \sm)$* [Corollary \ref{CO2}]&
$ O(\log \sm)$* [Corollary \ref{CO2}] \\
\hline
all- & honest & \multirow{2}{*}{$\Theta(\sm)$} 
 &$ \Theta(\sm)$ [Theorem \ref{theo:multiround}] &
 \multirow{2}{*}{$\Theta(\sm)$}& $O(\log \sm)$* [Corollary \ref{CO3}] \\
 \cline{2-2}\cline{4-4} \cline{6-6}
round & specious & 
\multirow{2}{*}{implied by \cite{BB15}}
 &$ \Theta(\sm)$ implied by \cite{BB15}& \multirow{2}{*}{implied by \cite{ABCGLS19}}
 & $ \Theta(\sm)$ implied by \cite{ABCGLS19}\\
 \cline{1-2} \cline{4-4} \cline{6-6}
 \multicolumn{2}{|c|}{dishonest} & 
  &$ \Theta(\sm)$ implied by \cite{BB15}&& $ \Theta(\sm)$ implied by \cite{ABCGLS19}\\
\hline
\end{tabular}
\end{center}
\begin{flushleft}
This table employs the same notations as Table \ref{tab:existing-results}.
\end{flushleft}
\end{table*}

\subsection{Outline of results and comparison}
\subsubsection{Optimality of trivial solution for one-server Q-QPIR}
First, we discuss our result for one-server
Q-QPIR for the honest-server model without prior entanglement,
and its relation to existing results.
The result by the reference \cite{BB15} is summarized as follows.
The C-QPIR protocol discussed in \cite{BB15}
is considered as a QPIR protocol over the set ${\cal C}$.
The reference \cite{BB15} showed that 
the trivial protocol over the set ${\cal C}$
is optimal in the all-round criterion
under the specious-server model with input states ${\cal Q}$, 
i.e., when the secrecy in the all-round criterion
is imposed under the specious-server model with input states ${\cal Q}$.
Since the set ${\cal C}=\{ |j\rangle\}_{j=0}^{d-1}$ is 
included in the set ${\cal Q}$,
a Q-QPIR protocol over the set $\mathcal{Q}$ 
works a QPIR protocol over the set ${\cal C}$.
Hence, the result by \cite{BB15} implies the optimality of the trivial protocol
over the set ${\cal Q}$
 in the all-round criterion under the specious-server model.
In addition, such an impossibility result was extended to the case with 
prior entanglement by the paper \cite{ABCGLS19}.

However, the secrecy in the all-round criterion
under the specious-server model 
is a stronger condition than 
the secrecy in the final-state criterion
under the honest-server model
because the secrecy in the all-round criterion
is a stronger condition the secrecy in the final-state criterion
and the specious-server model allows the server to have a larger choice
than the honest-server model. 

To seek further possibility for C-QPIR protocols, 
in Sections \ref{sec-one} and \ref{sec-one2},
inspired by the idea presented in \cite{GLM},
we propose more efficient one-round C-QPIR protocols 
in the final-state criterion under the honest-server and specious-server models
with input states ${\cal C}$ or ${\cal Q}$ 
whose communication complexities are at most $4\log \sm$.
In addition, the reference \cite{LeG12} proposed 
a C-QPIR protocol in the all-round criterion
under the honest one-server model
	that has communication complexity $O(\sqrt{\sm})$.
The reference \cite{KLLGR16} also proposed 
a C-QPIR protocol with communication complexity
$O(\mathrm{poly} \log \sm)$
without prior entanglement and
a C-QPIR protocol with communication complexity
$O( \log \sm)$ with prior entanglement.
In Section \ref{sec-all}, we show that
these two protocols satisfy the secrecy in the all-round criterion
under the honest-server model with input states ${\cal C}$.
In addition, using a conversion result \cite{ABCGLS19},
we show that these two protocols satisfy the secrecy in the all-round criterion
under the specious-server model with input states ${\cal C}$.

Hence,	we cannot exclude the possibility of more efficient one-server Q-QPIR protocols than 
the trivial solution in the final-state criterion or under the honest one-server model.
Furthermore, while the trivial solution is optimal under the honest-server model of classical PIR \cite{CGKS98}, 
its optimality proof uses the communication transcript between the server and the user, which is based on classical communication.
Unfortunately, we cannot apply the same technique under the honest one-server model of 
Q-QPIR because quantum states cannot be copied because of the no-cloning theorem.
Therefore, we have a question of whether there exists 
a Q-QPIR protocol over pure states
that satisfies the secrecy in the final-state criterion under the honest-server model, and improves the communication complexity over the trivial protocol.

As its solution, we show that the trivial solution is optimal for one-server 
Q-QPIR in the final-state criterion for the honest-server model.
In Tables~\ref{tab:existing-results} and \ref{tab:blind}, 
we summarize the comparison of our results with previous results for the one-server case. 
In our proof, the entropic inequalities are the key instruments for the proof.
Since the pair of the final-state criterion and the honest-server model 
is the weakest attack model, this result implies that the trivial solution is also optimal for any attack model.

\subsubsection{One-server Q-QPIR protocol with prior entanglement}
However, the above discussion assumes that there is no prior entanglement
shared between the sender and the user.
Hence, secondly, with prior entanglement between the user and the server,
we prove that there exists an efficient Q-QPIR protocol on the honest-server model or on
the final-state criterion.
To be precise, we propose a method to construct a Q-QPIR protocol of communication complexity $O(f(\sm))$ with prior entanglement from a C-QPIR protocol of communication complexity $O(f(\sm))$ with prior entanglement.
This method is based on the combination of 
C-QPIR and quantum teleportation \cite{BBCJPW93}.
The proposed Q-QPIR protocol inherits the security of the C-QPIR protocol.
With this property, we show three types of Q-QPIR protocols of communication complexity $O(\log \sm)$ with prior entanglement.
One is the secrecy in the final-state criterion under the honest-server model.
The second is the secrecy in the final-state criterion under the specious-server model.
The third is the secrecy in the all-round criterion under the honest-server model.
Combining this result with the above result, we find that prior entanglement realizes an exponential speedup
for one-server Q-QPIR 
in the final-state criterion or under the honest-server model.
Therefore, the obtained results are summarized as Table \ref{tab:existing-results} 
in terms of the communication complexity $\sm$. 

\section{Preliminaries} \Label{sec:prelim}
We define $[a:b]  = \{a,a+1, \ldots, b\}$ and $[a] = \{1,\ldots, a\}$.
The dimension of a quantum system $X$ is denoted by $|X|$. 
The von Neumann entropy is defined as $H(X) = H(\rho_X) = \Tr \rho_X\log \rho_X$, where $\rho_X$ is the state on the quantum system $X$.

\begin{prop}\Label{PP3}
The von Neumann entropy satisfies the following properties.

 \noindent$(a)$ $H(X) = H(Y)$ if the state on $X\otimes Y$ is a pure state.
 
 \noindent$(b)$ The inequality $H(XY) \ge H(X) + H(Y)$ holds, and the equality holds for product states on $X\otimes Y$.
 
 \noindent$(c)$ Entropy does not change by unitary operations.
 
 \noindent$(d)$ $H(XY) + H(X) \geq H(Y)$.
 
 \noindent$(e)$ $H(\sum_s p_s \rho_s) = \sum_s p_s (H( \rho_s) - \log p_s)$ if $\Tr\rho_s \rho_t = 0$ for any $s\neq t$.
\end{prop}
 
\noindent

The property $(d)$ is proved as follows.
Since other properties can be easily shown, we omit their proofs.
For example, see the book \cite[Sections 3.1 and 8.1]{Hayashi}. 
Let $Z$ be the reference system in which the state on $XYZ$ is pure. 
Then, $H(XY)+H(X)=H(Z)+H(X)\ge H(XZ)=H(Y)$.
Throughout the paper, we use the symbols $(a)$, $(b)$, $(c)$, $(d)$, $(e)$ to denote which property is used, e.g., $\stackrel{\mathclap{(a)}}{=}$ means that the equality holds from the property $(a)$.

Next, for a TP-CP map from the system ${\cal H}_X$ to the system ${\cal H}_Y$ and a state $\rho$
on ${\cal H}_X$,
we define the transmission information $I (\rho,\Gamma)$. 
We choose a purification $|\psi\rangle$ of $\rho$ with the environment ${\cal H}_Z$.
Then, the transmission information $I (\rho,\Gamma)$ is defined as
\begin{align}
I (\rho,\Gamma):=
H(\rho)+ H(\Gamma(\rho))-H( (\iota_Z\otimes \Gamma)(|\psi\rangle \langle \psi|)  ),
\end{align}
where $\iota_Z$ is the identity operation on ${\cal H}_Z$.
When  $\Gamma$ is the identity operator,
\begin{align}
I (\rho,\Gamma)=2H(\rho).\Label{MMX}
\end{align}

Throughout this paper, 
$\mathbb{C}^d$ expresses the $d$-dimensional Hilbert space
spanned by the orthogonal basis $\{|s\rangle\}_{s=0}^{d-1}$.
For a $d_1\times d_2$ matrix 
	\begin{align}
	\sM = \sum_{s=0}^{d_1-1} \sum_{t=0}^{d_2-1} m_{st} |s\rangle\langle t|  \in \mathbb{C}^{d_1\times d_2},
	\end{align}
	we define 
\begin{align}
|\sM\rrangle = \frac{1}{\sqrt{d}}\sum_{s=0}^{d_1-1} \sum_{t=0}^{d_2-1} m_{st} |s\rangle| t\rangle  \in \mathbb{C}^{d_1}\otimes \mathbb{C}^{d_2}.
\end{align}
For $\sA \in \mathbb{C}^{d_1\times d_2}$, $\sB\in\mathbb{C}^{d_1 \times d_1}$, and $\sC \in\mathbb{C}^{d_2 \times d_2}$,
	we have the relation 
\begin{align}
(\sB\otimes \sC^{\top}) | \sA \rrangle =  | \sB\sA\sC \rrangle.
\end{align}

We call a $d$-dimensional system $\mathbb{C}^d$ a {\em qudit}.
Define generalized Pauli matrices and the maximally entangled state on qudits as 
\begin{align}
\sX_d &= \sum_{s=0}^{d-1} |s+1\rangle \langle s|,\\ 
\sZ_d &= \sum_{s=0}^{d-1} \omega^{s} |s\rangle \langle s|,\\
|\sI_d \rrangle &= \frac{1}{\sqrt{d}} \sum_{s=0}^{d-1} |s,s\rangle,
\Label{eq:defs-pauli}
\end{align}
where $\omega = \exp(2\pi\ii/ d)$ and $\iota = \sqrt{-1}$.
We define 
	the generalized Bell measurements
	\begin{align}
		\mathbf{M}_{\sX\sZ,d} = \{ |\sX^a \sZ^b\rrangle \mid a,b \in [0:d-1] \}. \Label{mes}
	\end{align}
If there is no confusion, we denote $\sX_d,\sZ_d, \sI_d, \mathbf{M}_{\sX\sZ,d}$ by $\sX,\sZ, \sI, \mathbf{M}_{\sX\sZ}$.
Let $A, A',B, B'$ be qudits.
If the state on $A\otimes A' \otimes B \otimes B'$ is $|\sA\rrangle\otimes|\sB\rrangle$ 
	and the measurement $\mathbf{M}_{\sX\sZ}$
	is performed on $A' \otimes B'$ with outcome $(a,b) \in [0:d-1]^2$,
	the resultant state is 
	\begin{align}
	|\sA\sX^a\sZ^{-b} \sB^\top\rrangle \in A \otimes B.
		\Label{qe:feafterf}
	\end{align}

We also define the dual basis
\begin{align}
|u_{j}\rangle:=\sum_{k=0}^{d-1}\frac{1}{\sqrt{d}}e^{\frac{2\pi kj i}{d}}|k\rangle.\Label{E12}
\end{align}

\section{Protocols for C-QPIR} \Label{C-QPIR-pro}
\subsection{One-round C-QPIR of the final-state criterion
under honest-server model} \Label{sec-one}
This section presents a
protocol that satisfies the secrecy in the final-state criterion
under the honest-server model with the input states ${\cal C}$.
We assume that the $\ell$-th message $X_\ell$ is an element of $\bZ_{d_\ell}$
for $\ell \in [\sff]$.
We define $d$ as the maximum $\max_{\ell\in [\sff]}d_\ell$.
\begin{prot}
\Label{PR2}
The following protocol is denoted by
$\Phi_{\sff,d}$.
\begin{description}[leftmargin=1.5em]
\item[0)] \textbf{Preparation}: 
The server prepares $\sff+1$ quantum systems ${\cal H}_0,{\cal H}_1, \ldots, {\cal H}_{\sff}$, where 
${\cal H}_{0}$ is spanned by $\{|j\rangle\}_{j=0}^{d-1} $, and 
${\cal H}_{\ell}$ is spanned by $\{|j\rangle\}_{j=0}^{d_\ell-1} $. 
When the $\ell$-th message is $X_\ell$,
the state on the quantum system ${\cal H}_\ell$ is set to be 
$|X_\ell\rangle $.
Also, the state on the quantum system ${\cal H}_0$ is set to be 
$|0 \rangle $.
The user prepares the system ${\cal K}$ 
spanned by $\{ |\ell \rangle \}_{\ell=1}^{\sff}$.
\item[1)] \textbf{Query}: 
The user sets the state on the system ${\cal K}$ to be $|K\rangle$.
The user sends the system ${\cal K}$ to the server.

\item[2)] \textbf{Answer}: 
The server applies the measurement based on the computation basis
$\{ |j\rangle \}$ on the systems ${\cal H}_1, \ldots, {\cal H}_{\sff}$
with the projective state reduction.
The server applies the controlled unitary $U:= \sum_{\ell=1}^{\sff}
|\ell \rangle \langle \ell| \otimes U_\ell$
on ${\cal K}\otimes {\cal H}_0\otimes {\cal H}_1 \otimes \cdots \otimes {\cal H}_\sff$, where
$U_\ell$ acts only on ${\cal H}_0\otimes {\cal H}_\ell$ and is defined as
\begin{align}
U_\ell:=\sum_{j'=0}^{d-1}\sum_{j=0}^{d_\ell-1}
|j+j'\rangle \langle j'|\otimes |j\rangle \langle j|.
\end{align}
The server sends the system ${\cal K}\otimes {\cal H}_0$
to the user.
\item[3)] \textbf{Reconstruction}:
The user measures ${\cal H}_0$, and obtains the message $X_{K}$.
\end{description}
\end{prot}

The correctness of Protocol \ref{PR2} is trivial.
Its upload and download complexities are 
$UC(\Phi_{\sff,d})=\log \sff$ and 
$DC(\Phi_{\sff,d})=\log \sff+ \log d$.
The communication complexity is 
$CC(\Phi_{\sff,d})=2 \log \sff+ \log d$.
When $d$ is fixed, $CC(\Phi_{\sff,d})=2\log \sm+o(\sm)$.

As shown in the following; Protocol \ref{PR2} satisfies the secrecy in the final-state criterion
under the honest-server model with the input states ${\cal C}$.
We assume that the server and the user are honest.
Since the server follows the protocol, 
the server has only the $\sff$ systems
${\cal H}_1, \ldots, {\cal H}_{\sff}$.
The final state on the composite system 
${\cal H}_1\otimes \ldots\otimes {\cal H}_{\sff}$
is $|X_1\rangle\cdots |X_\sff\rangle$,
which does not depend on the user's choice $K$.
Hence, the above secrecy holds.

In addition, Protocol \ref{PR2} satisfies the secrecy in the final-state criterion
under the honest-server model even with the input states ${\cal Q}$ as follows.
Even when the initial states in ${\cal H}_1, \ldots, {\cal H}_{\sff}$ prepared as quantum states,
due to the measurement, the initial states in ${\cal H}_1, \ldots, {\cal H}_{\sff}$ 
are convex mixtures of states $\{|j\rangle\langle j|\}$.
Hence, the final state on the composite system ${\cal H}_1\otimes \ldots\otimes {\cal H}_{\sff}$
is the same as the state after the measurement,
which does not depend on user's choice $K$.
Hence, the above secrecy holds.

However, when the server skips the measurement in Step 2) and 
the input states are chosen as ${\cal Q}$,
the secrecy does not hold as follows.
Assume that the server set initial state in ${\cal H}_{\ell} $ to be
$\sum_{j=1}^{d_\ell}\frac{1}{\sqrt{d_\ell}}|j\rangle$.
Also, we assume that the server and the user follow Steps 1), 2), 3).
Then, the final state on ${\cal H}_{K} \otimes {\cal H}_0$ is 
$\sum_{j=1}^{d_\ell}\frac{1}{\sqrt{d_\ell}}|j\rangle|j\rangle$.
That is, the final state on ${\cal H}_{K}$ is the completely mixed state.
In contrast, the final state on ${\cal H}_{\ell}$ is the same as the initial state 
for $\ell \neq K$.
Hence, the secrecy condition \eqref{MLK} does not hold.

Also, Protocol \ref{PR2} does not have the secrecy with the input states ${\cal C}$
under the specious-server model as follows.
A specious server is allowed to make a measurement 
if the measurement does not destroy the quantum state.
Since the state on the composite system 
${\cal K}\otimes {\cal H}_0\otimes {\cal H}_1 \otimes \cdots 
\otimes {\cal H}_\sff$
is one of the basis, it is not destroyed by the basis measurement. 
Hence, the server can obtain the user's choice $K$ without state demolition.
This fact shows that the specious-server model is needed in order 
to forbid such an insecure protocol.
However, as shown in Section \ref{sec:blind-one-opt},
even under the honest-server model,
a protocol similar to Protocol \ref{PR2} does not work
when the messages are given as quantum states.

\subsection{One-round C-QPIR of the final-state criterion under specious-server model} \Label{sec-one2}
Protocol \ref{PR2} presented in the previous subsection
does not work under the specious-server model.
To resolve this problem, this section presents a
protocol that satisfies the secrecy in the final-state criterion
under the specious-server model with the input states ${\cal C}$.
We assume that each message $X_\ell$ is an element of $\bZ_{d_\ell}$.
We define $d$ as the maximum $\max_{\ell}d_\ell$.
\begin{prot}
\Label{PR2B}
The following protocol is denoted by
$\Phi_{\sff,d}$.
\begin{description}[leftmargin=1.5em]
\item[0)] \textbf{Preparation}: 
The server prepares $\sff+2$ quantum systems 
${\cal H}_{0}',{\cal H}_{1}',{\cal H}_1, \ldots, {\cal H}_{\sff}$, where 
${\cal H}_{0}',{\cal H}_{1}'$ is spanned by $\{|j\rangle\}_{j=0}^{d-1} $, and 
${\cal H}_{\ell}$ is spanned by $\{|j\rangle\}_{j=0}^{d_\ell-1} $. 
When the $\ell$-th message is $X_\ell$,
the state on the quantum system ${\cal H}_\ell$ is set to be 
$|X_\ell\rangle $.
Also, the state on the quantum system ${\cal H}_{0}',{\cal H}_1'$ is set to be 
$|0 \rangle $.
The user prepares the systems ${\cal K}_{0}$,${\cal K}_1$ 
spanned by $\{ |\ell \rangle \}_{\ell=1}^{\sff}$.
\item[1)] \textbf{Query}: 
The user generates the binary random variable $A$ 
and the variable $B \in [\sff]$ subject to the uniform distribution.
The user sets the state on the system ${\cal K}_A$ to be $|K\rangle$,
and the state on the system ${\cal K}_{A\oplus 1}$ to be 
$\frac{1}{\sqrt{\sff}}\sum_{\ell=1}^{\sff} \sZ_{\sff}^B |\ell\rangle$.
The user sends the systems ${\cal K}_0$,${\cal K}_1$ to the server.

\item[2)] \textbf{Answer}: 
The server applies the controlled unitary $U:= \sum_{\ell=1}^{\sff}
|\ell \rangle \langle \ell| \otimes U_\ell$
on ${\cal K}_0\otimes {\cal H}_0'\otimes {\cal H}_1 \otimes \cdots \otimes {\cal H}_\sff$, where
$U_\ell$ acts only on ${\cal H}_0'\otimes {\cal H}_\ell(={\cal H}_1'\otimes {\cal H}_\ell)
$ and is defined as
\begin{align}
U_\ell:=\sum_{j'=0}^{d-1}\sum_{j=0}^{d_\ell-1}
|j+j'\rangle \langle j'|\otimes |j\rangle \langle j|.
\end{align}
Then, the server applies the controlled unitary $U$
on ${\cal K}_1\otimes {\cal H}_1'\otimes {\cal H}_1 \otimes \cdots \otimes {\cal H}_\sff$.
The server sends the systems 
${\cal K}_0\otimes {\cal H}_0'$,
${\cal K}_1\otimes {\cal H}_1'$ to the user.

\item[3)] \textbf{Reconstruction}:
The user measures ${\cal H}_A'$, and obtains the message $X_{K}$.
\end{description}
\end{prot}

The correctness of Protocol \ref{PR2B} is trivial.
Its upload and download complexities are 
$UC(\Phi_{\sff,d})=2 \log \sff$ and 
$DC(\Phi_{\sff,d})=2\log \sff+ 2 \log d$.
The communication complexity is 
$CC(\Phi_{\sff,d})=4 \log \sff+ 2 \log d$.
When $d$ is fixed, $CC(\Phi_{\sff,d})=4\log \sm+o(\sm)$.

As shown below, Protocol \ref{PR2B} satisfies the secrecy in the final-state criterion
under the specious-server model with the input states ${\cal C}$.

Assume that the server and the user follow the protocol.
Then, the resultant state in the server's system
${\cal H}_1\otimes \ldots\otimes {\cal H}_{\sff}$
is the product state $|X_1\rangle  \ldots|X_\sff\rangle$.
The resultant state in 
${\cal K}_A\otimes {\cal H}_A'$ is
$|K\rangle |X_{K}\rangle $.
The resultant state in 
${\cal K}_{A\oplus 1}\otimes {\cal H}_{A\oplus 1}'$ is
$\frac{1}{\sqrt{\sff}}\sum_{\ell=1}^{\sff}\sZ_{\sff}^B|\ell\rangle |X_{\ell}\rangle $.

Hence, when $A=0$, the specious server needs to generate the state 
$|K\rangle |X_{K}\rangle\frac{1}{\sqrt{\sff}}
\sum_{\ell=1}^{\sff}\sZ_{\sff}^B|\ell\rangle |X_{\ell}\rangle $
from the state $|K\rangle\frac{1}{\sqrt{\sff}}\sum_{\ell=1}^{\sff}
\sZ_{\sff}^B|\ell\rangle $.
Also, when $A=1$, the specious server needs to generate the state 
$\frac{1}{\sqrt{\sff}}\sum_{\ell=1}^{\sff}\sZ_{\sff}^B|\ell\rangle |X_{\ell}\rangle |K\rangle |X_{K}\rangle$
from the state 
$\frac{1}{\sqrt{\sff}}\sum_{\ell=1}^{\sff}\sZ_{\sff}^B|\ell\rangle |K\rangle$.

Since the resultant states
$|K\rangle |X_{K}\rangle\frac{1}{\sqrt{\sff}}\sum_{\ell=1}^{\sff}
\sZ_{\sff}^B|\ell\rangle |X_{\ell}\rangle $
and
$\frac{1}{\sqrt{\sff}}
\sum_{\ell=1}^{\sff}\sZ_{\sff}^B|\ell\rangle |X_{\ell}\rangle |K\rangle |X_{K}\rangle
$
are unitarily equivalent to the states
$|K\rangle\frac{1}{\sqrt{\sff}}\sum_{\ell=1}^{\sff}\sZ_{\sff}^B|\ell\rangle $
and
$\frac{1}{\sqrt{\sff}}\sum_{\ell=1}^{\sff}\sZ_{\sff}^B|\ell\rangle |K\rangle$,
it is sufficient to discuss whether 
the server can get certain information from 
the state family
${\cal F}:=
\{
|k\rangle\frac{1}{\sqrt{\sff}}\sum_{\ell=1}^{\sff}
\sZ_{\sff}^b|\ell\rangle ,
\frac{1}{\sqrt{\sff}}\sum_{\ell=1}^{\sff}\sZ_{\sff}^b|\ell\rangle |k\rangle
\}_{k,b=1}^{\sff}$ without disturbance.
However, Koashi-Imoto \cite{KI,HJPW,BMPV,WSM} 
theory forbids the server 
to make any measurement when 
the states need to be recovered
because the state family ${\cal F}$ is composed of non-commutative states.
Therefore, when the server keeps the condition for the specious server,
the server cannot obtain any information for $K$.

However, it is not clear whether
adding the measurement in Step 2) guarantees that
the protocol satisfies the secrecy 
in the final-state criterion
under the specious-server model with the input states ${\cal Q}$.

\subsection{C-QPIR in all-round criterion} \Label{sec-all}
In this section we discuss the secrecy in the all-round criterion of the C-QPIR protocol with communication complexity 
$O(\mathrm{poly} \log \sm)$ 
under the fixed message size $d=2$
from \cite[Section 5]{KLLGR16}, which does not use any prior entanglement, and the C-QPIR protocol with communication $O( \log \sm)$ 
under the fixed message size $d=2$
from \cite[Section 6]{KLLGR16}, which uses $\Theta(\sm)$ ebits of prior entanglement.
Although these protocols fix the message size $d$ to be $2$,
they can be considered as protocols whose message sizes are fixed to an arbitrary $d$
by treating $\lceil \log_2 d \rceil$ messages as one message.
We first show the secrecy of the protocol from \cite[Section 5]{KLLGR16} under the honest server model.

\begin{lemm}\Label{prop:sec1}
The protocol from \cite[Section 5]{KLLGR16} 
is unitary-type and satisfies the secrecy
in the all-round criterion under the honest server model when the set $\tilde{\cal S}$ of possible inputs is ${\cal C}$.
\end{lemm}
\begin{proof}
The protocol from \cite[Section 5]{KLLGR16} works for the case $d=2$. The server's input is thus $(a_1,\ldots,a_\sff)$ for $a_1,\ldots,a_\sff\in\{0,1\}$. The user's input is an index $K\in\{1,\ldots,\sff\}$.

The main idea is to simulate a classical multi-server PIR protocol with $s=O(\log \sm)$ servers that has total communication complexity $O(\mathrm{poly} \log \sm)$. 
Such protocols are known to exist (see, e.g., \cite{{CGKS98}}) and can be described generically as follows. 
The user picks a uniform random variable $G$ from $\{1,\ldots,\sg\}$, 
computes an $s$-tuple of queries 
$\{q_1(G,K),\ldots,q_s(G,K)\}$ from $(G,K)$ by using a function $q_t$, and asks query 
$q_t(G,K)$ to the $t$-th server. 
Here, for each $t\in\{1,\ldots, s\}$, the function $q_t$ satisfies the condition that 
the distribution of query $q_t(G,K)$ is independent of $K$.
Each server $t$ then sends its answer $\fans_t(q_t(G,K))$ 
to the user, who recovers $a_K$ 
from $\{\fans_1(q_1(G,K)),\ldots,\fans_s(q_s(G,K))\}$. 

The protocol from \cite[Section 5]{KLLGR16} simulates this protocol using only one server. 
The protocol uses $2s+1$ quantum registers denoted $\reQ,\reQ_1,\ldots, \reQ_s, \reAns_1,\ldots, \reAns_s$. For each $t\in\{1,\ldots, s\}$, let us define the following quantum state:
\begin{align*}
&\ket{\Phi_t}\\
=&\frac{1}{\sqrt{\sg}}\sum_{g}\ket{q_1(g,K),\cdots,q_s(g,K)}_{\reQ}
\ket{q_1(g,K)}_{\reQ_1}\cdots\ket{q_s(g,K)}_{\reQ_s}
\\
&\otimes\ket{\fans_1(q_1(g,K))}_{\reAns_1}\cdots
\ket{\fans_{t-1}(q_{t-1}(g,K))}_{\reAns_{t-1}} \\
&\otimes\ket{0}_{\reAns_{t}}\cdots\ket{0}_{\reAns_{s}}.
\end{align*}
Note that we have in particular
\begin{align*}
&\ket{\Phi_1}\\
=&\frac{1}{\sqrt{\sg}}\sum_{g}
\ket{q_1(g,K),\cdots,q_s(g,K)}_{\reQ}
\ket{q_1(g,K)}_{\reQ_1}\cdots\ket{q_s(g,K)}_{\reQ_s}\\
&
\otimes \ket{0}_{\reAns_1}\cdots\ket{0}_{\reAns_{s}}.
\end{align*}
The protocol consists of the following interaction between the user and the server (some details of the manipulations of the states are omitted since they are irrelevant to the secrecy proof):
\begin{itemize} 
\item[1.]
The user prepares the state $\ket{\Phi_1}$.
\item[2.]
The user and the server iterate the following for $t=1$ to $s$:
\begin{itemize}
\item[2.1]
The user sends Registers $\reQ_{t},\reAns_{t}$ to the server;
\item[2.2]
The server applies a controlled unitary, 
where the controlling system is $\reQ_{t}$
and 
the controlled system is $\reAns_{t}$.
Then, the server sends back Registers $\reQ_{t},\reAns_{t}$ to the user.
\end{itemize}
\item[3.] 
The user measures the joint system composed of
Registers $\reQ,\reQ_1,\ldots, \reQ_s, \reAns_1,\ldots, \reAns_s$
to obtain the outcome $a_K$
after certain unitary operations.
\end{itemize}
Since this protocol is unitary-type,
the remaining task is to show the secrecy of this protocol in the all-round criterion under the honest server model when the set $\tilde{\cal S}$ of possible inputs is~${\cal C}$. 
Observe that at each iteration there is only a message sent to the server, at Step~2.1. We thus only need to show that for each $t$, this message does not reveal any information about $K$. The state of the whole system at the end of Step~2.1 of the $t$-th iteration is $\ket{\Phi_t}$. The state of the server, obtained by tracing out all registers except $\reQ_{t},\reAns_{t}$ of $\ket{\Phi_t}\bra{\Phi_t}$ is
\begin{align}
\frac{1}{\sg}
\sum_{g}\ket{q_t(g,K)}_{\reQ_t}\ket{0}_{\reAns_t}
\bra{q_t(g,K)}_{\reQ_t}\bra{0}_{\reAns_t}.
\end{align}
Since the distribution of query $q_t(G,K)$ is independent of $K$, we conclude that the whole state of the server at the end of Step 2.1 is independent of $K$, for each $t$.
\end{proof}

Next, we show the secrecy of the protocol from \cite[Section 6]{KLLGR16} under the honest server model (see also Appendix B in \cite{ABCGLS19}).
\begin{lemm}\Label{prop:sec2}
The protocol from \cite[Section 6]{KLLGR16} 
is unitary-type and satisfies the secrecy
in the all-round criterion under the honest server model when the set $\tilde{\cal S}$ of possible inputs is ${\cal C}$.
\end{lemm}
\begin{proof}
The protocol from \cite[Section 6]{KLLGR16} works for the case $d=2$ and $\sff=2^\sh$, for $\sh\ge 1$.
The server's input is thus $(a_1,\ldots,a_\sff)$ for $a_1,\ldots,a_\sff\in\{0,1\}$. The user's input is an index $K\in\{1,\ldots,\sff\}$.

The protocol uses $2\sh+2$ quantum registers denoted $\reR_1,\ldots, \reR_\sh, \mathsf{R'}_1,\ldots, \mathsf{R'}_\sh, \reQ_0,\reQ_1$. 
For each $p\in\{1,\ldots, \sh\}$, let us define the following quantum state over the two registers $\reR_t, \reR'_p$:
\[
\ket{\Phi_p}=\frac{1}{\sqrt{2^{2^{\sh-p}}}}
\sum_{z\in \{0,1\}^{2^{\sh-p}}}\ket{z}_{\reR_p}\ket{z}_{\mathsf{R'}_p}.
\]
For any binary string $z\in\{0,1\}^s$ with $s$ even, we denote $z[0]$ the first half of $z$, and $z[1]$ the second half of $z$. For any binary strings $z, z'\in\{0,1\}^s$, we write $z\oplus z'\in\{0,1\}^s$ the string obtained by taking the bitwise parity of $z$ and~$z'$.

The protocol from \cite[Section 6]{KLLGR16} assumes that the server and the user initially share the state
\[
\ket{\Phi_1}\otimes \cdots\otimes\ket{\Phi_\sh}\cdots \otimes\ket{0}_{\reQ_0}\ket{0}_{\reQ'_0},
\]
where $\reR_1,\ldots, \reR_\sh,\reQ_0, \reQ_1$ are owned by the server and 
$\reR'_1,\ldots, \reR'_\sh$ are owned by the user. The protocol consists of the following interaction between the user and the server (some details of the manipulations of the states are omitted since they are irrelevant to the secrecy proof):
\begin{itemize} 
\item[1.]
For $p$ from 1 to $\sh$ the server and the user do the following:
\begin{itemize}
\item[1.1]
The server applies a unitary $V_p$ (defined in \cite[Eq.~(27)]{KLLGR16}) on Registers $\reR_{p-1},\reR_p$, $\reQ_0, \reQ_1$ and then sends Registers $\reQ_0, \reQ_1$ to the user;
\item[1.2]
If the $p$-th bit of its input $K$ is $0$, the user applies the Pauli gate $Z$ on Register $\reQ_0$. If the $p$-th bit of $K$ is $1$, the user applies $Z$ on Register $\reQ_1$. The user then sends back Registers $\reQ_0, \reQ_1$ to the server.
\item[1.3]
The server applies again the unitary $V_p$ on Registers $\reR_{p-1},\reR_p$, $\reQ_0, \reQ_1$, and then applies a Hadamard transform on each qubit in Register $\reR_p$.  
\item[1.4]
The user applies a Hadamard transform on each qubit in Register $\reR'_p$.
\end{itemize}
\item[2.] 
The server sends Register $\reR_\sh$ to the user. The user measures the joint system composed of
Registers $\reR'_1,\ldots, \mathsf{R'}_\sh$ and Register $\reR_\sh$, and performs some classical post-processing on the outcome to obtain $a_K$
\end{itemize}
Since this protocol is unitary-type,
the remaining task is to show the secrecy of this protocol in the all-round criterion under the honest server model when the set $\tilde{\cal S}$ of possible inputs is~${\cal C}$. 
Since the initial state does not depend on $K$, 
it is sufficient to show that 
the whole state on Register $\reR_1,\ldots, \reR_\sh,\reQ_0, \reQ_1$
at the end of Step 1.2 of the $p$-th round is independent of $K$.

Lemma~2 in~\cite{KLLGR16} shows that the state of the whole system at the end of Step~1.3 is
\[
\ket{\Psi_p}\otimes
\bigotimes_{j=p+1}^\sh \ket{\Phi_j}_{(\reR_j,\reR'_j)}\\
\otimes 
\ket{0}_{\reQ_0}\ket{0}_{\reQ'_0}
\]
with
\[
\ket{\Psi_p}=
\frac{1}{\sqrt{2^{2^{\sh-1}}\cdots 2^{2^{\sh-p}}}}
\sum_{y^1,\ldots,y^p}
\bigotimes_{j=1}^p \ket{y^j}_{\reR_j}\ket{y^{j-1}[i_j]\oplus y^j}_{\reR'_j},
\]
where the sum is over all strings $y^1\in \{0,1\}^{2^{\sh-1}},\ldots,y^p\in \{0,1\}^{2^{\sh-p}}$ and we use the convention that $y^0$ is the server's input $(a_1,\ldots,a_\sff)$.\footnote{Observe that $y^{j-1}$ is a binary string of length $2^{\sh-(j-1)}$, and then $y^{j-1}[i_j]$ is a binary string of length $2^{\sh-(j-1)-1}=2^{\sh-j}$. The term $y^{j-1}[i_j]\oplus y^j$ in the definition of $\ket{\Psi_p}$ is thus well defined.} Here the server owns Registers $\reR_1,\ldots, \reR_\sh,\reQ_0, \reQ_1$ while the user owns Registers $\reR'_1,\ldots, \reR'_\sh$. 
Observing that tracing out Registers $\reR'_1,\ldots,\reR'_j$ from 
$\ket{\Psi_p}\bra{\Psi_p}$ gives the state
\[
\frac{1}{2^{2^{\sh-1}}\cdots 2^{2^{\sh-p}}}
\sum_{y^1,\ldots,y^p}\ket{y^1}_{\reR_1}\cdots\ket{y^p}_{\reR_p}
\bra{y^1}_{\reR_1}\cdots\bra{y^p}_{\reR_p},
\]
which is independent of $K$,
we find that 
the whole state on
Register $\reR_1,\ldots, \reR_\sh,\reQ_0, \reQ_1$
at the end of Step 1.3 of the $p$-th round
is independent of $K$, for each $p$.
Since the unitaries applied in Step 1.3 by the server are independent of $K$,
we conclude that 
the whole state on 
Register $\reR_1,\ldots, \reR_\sh,\reQ_0, \reQ_1$
at the end of Step 1.2 of the $p$-th round is independent of $K$.
\end{proof}

Finally, we discuss the secrecy under the specious server model. 
We will rely on the following theorem from~\cite{ABCGLS19} for unitary-type QPIR protocols. 

\begin{prop}[Theorem 3.2 in \cite{ABCGLS19}]
When a unitary-type QPIR protocol satisfies the secrecy in the all-round criterion under the honest server model with the set $\tilde{\cal S}=\cal C$,
it 
satisfies the secrecy
in the all-round criterion under the specious server model with the same set $\tilde{\cal S}=\cal C$.
\end{prop}

Therefore, we obtain the following corollary of Lemmas \ref{prop:sec1} and \ref{prop:sec2}.

\begin{coro}\Label{CCor1}
The protocols from \cite[Section 5]{KLLGR16} and \cite[Section~6]{KLLGR16} 
satisfy the secrecy
in the all-round criterion under the specious server model when the set $\tilde{\cal S}$ of possible inputs is ${\cal C}$.
\end{coro}

Therefore, 
when the message size $d$ is fixed to a constant,
there exists a C-QPIR protocol with communication complexity 
$O(\mathrm{poly} \log \sm)$ ($O( \log \sm)$)
and without any prior entanglement (with prior entanglement)
that satisfies the secrecy in the all-round criterion 
under the specious server model when the set $\tilde{\cal S}$ of possible inputs is ${\cal C}$.

\section{Optimality of trivial protocol in final-state criterion for Q-QPIR under honest server model} \Label{sec:blind-one-opt}

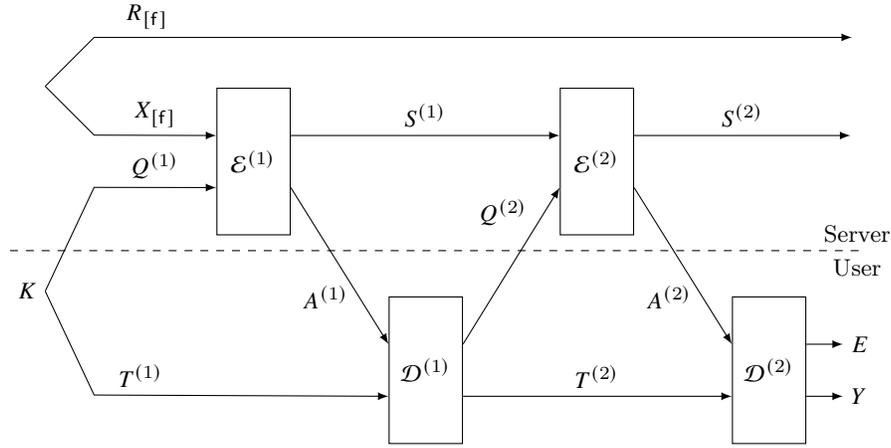
\begin{figure*}[t]
\begin{center}
\begin{tikzpicture}[scale=0.5, node distance = 3.3cm, every text node part/.style={align=center}, auto]
	\node [block, minimum height=6em] (enc1) {$\mathcal{E}^{(1)}$};
	\node [block, minimum height=6em, right=11em of enc1] (enc2) {$\mathcal{E}^{(2)}$};
	\node [block, minimum height=6em, below right=2.5em and 4em of enc1] (dec1) {$\cD^{(1)}$};
	\node [block, minimum height=6em, right=11em of dec1] (dec2) {$\cD^{(2)}$};

	\node [left=7em of enc1.145] (start) {\phantom{$X_{[\sff]}$}};
	\node [above=0.3em of start] (XX) {\phantom{$X_{[\sff]}$}};
	\path [line] (XX.east) --++(4em, 4em) -- node[above, pos=0.07] {$R_{[\sff]}$} ++(0:62em);
	\path [line] (XX.east) --++(4em,-4em) --node[above]{$X_{[\sff]}$}  (start -| enc1.west);
	
	\node [below left=5.7em and 7em of enc1.145] (K)  {$K$};
	\path [line] (K.east) --++(4em, 8.5em) -- node[above] {$Q^{(1)}$} (enc1.215);
	\path [line] (K.east) --++(4em,-8.5em) --node[above, pos=0.16]{$T^{(1)}$}  (dec1.215);

	\node [below left=3.5em and -1em of start] (lefts) {};
	\node [right=35em of lefts] (rights) {};
	\path [line,-,dashed] (lefts) -- node[below,pos=0.99]{User} node[above,pos=0.99]{Server} (rights);

	\path [line] (enc1.35) --node[above]{$S^{(1)}$} (enc2.145);
	\path [line] (enc1.325) -- node[below,pos=0.7,left=0.1em]{$A^{(1)}$} (dec1.145);
	\path [line] (enc2.325) -- node[below,pos=0.7,left=0.1em]{$A^{(2)}$} (dec2.145);

	\path [line] (dec1.35) --node[above,above left=1.3em and -1em]{$Q^{(2)}$} (enc2.215);
	\path [line] (dec1.325) -- node[above]{$T^{(2)}$} (dec2.215);

	\path [line] (enc2.35) --node[above]{$S^{(2)}$} ++(0:17.7em);
	\path [line] (dec2.325) -- ++(0:3em) node[right] {$Y$};
	\path [line] (dec2.35) -- ++(0:3em) node[right] {$E$};

\end{tikzpicture}
\caption{$2$-round QPIR protocol.}		\Label{fig:flow}
\end{center}
\end{figure*}

In this section, we prove that the trivial solution of downloading all messages is optimal 
		for Q-QPIR.
		In particular, this section, unlike the references \cite{BB15,ABCGLS19}, 
		we show the optimality in the final-state criterion under the honest-server model.
Since our setting is discussed under the honest-server model,
the secrecy in the final-state criterion is required 
only when
the server follows the determined state preparation process and determined quantum operations.
In the formal description of our protocols, we consider that the user and the server apply CPTP maps 
    but we describe the CPTP maps by the equivalent representation with the unitary maps and the local quantum memories.

To be precise, we define the $\sr$-round Q-QPIR protocol as follows. 
A $2$-round protocol is depicted in Figure~\ref{fig:flow}, and the symbols are summarized as Table \ref{tab:symbols}.
The message states are given as arbitrary $\sff$ states 
	$\rho_{[\sff]}:=\rho_1\otimes \cdots\otimes \rho_{\sff}$ on $S^{(0)} = X_{1}\otimes \cdots \otimes X_{\sff}$, where each of $\rho_\ell$ is purified in $X_\ell\otimes  R_\ell$. 
We use the notation $R_{[\sff]}:=R_1\otimes \cdots \otimes R_{\sff}$. 
The server contains the system $S^{(0)}$.
The user chooses the index of the targeted message $K\in[\sff]$, i.e., $\rho_k$ is the targeted quantum state when $K=k$.
When $K=k$, the user prepares the initial state as $|k\rangle \otimes |0\rangle \in A^{(0)} \otimes T^{(0)}$.
Although we consider the model in which the user and the server apply CPTP maps, we describe it by the equivalent representation with the unitary maps and the local quantum memories.
A Q-QPIR protocol $\Phi$ is described by unitary maps $\cD^{(0)},\ldots,\cD^{(\sr)}, \cE^{(1)},\ldots,\cE^{(\sr)}$ in the following steps.

\begin{table}[t]
\begin{center}
\caption{Definition of symbols} \Label{tab:symbols}
\begin{tabular}{|c|c|c|c|c|}
\hline
Symbol &	Definition	\\
\hline
$\sm$	&	Total size of messages (states)	\\
\hline
$\sff$	&	Number of messages (states)	\\
\hline
$\sr$	&	Number of rounds in multi-round models	\\
\hline
\end{tabular}
\end{center}
\end{table}

\begin{enumerate}[leftmargin=1.5em]
\item \textbf{Query}: 
For all $i \in [\sr]$, the user applies a unitary map $\cD^{(i-1)}$ from $A^{(i-1)}\otimes T^{(i-1)}$ to $Q^{(i)} \otimes T^{(i)}$,
and sends $Q^{(i)}$ to the sender.
Here, $T^{(i)}$ are the user's local quantum systems for describing the CPTP maps applied by the user.

\item \textbf{Answer}:
For all $i \in [\sr]$, the server applies a unitary map $\mathcal{E}^{(i)}$ from $Q^{(i)} \otimes S^{(i-1)}$ to $A^{(i)}\otimes S^{(i)} $
and sends $A^{(i)}$ to the user.
Here, $S^{(i)}$ are the server's local quantum systems for describing the CPTP maps applied by the server.

\item \textbf{Reconstruction}:
	The user applies $\cD^{(\sr)}$ from $A^{(\sr)}\otimes T^{(\sr)}$ to $Y \otimes E$, 
	and outputs the state on $Y$ as the protocol output.
	\end{enumerate}

The input-output relation $\Lambda_{\Phi}$ of the protocol $\Phi$ is written with a CPTP $\Gamma_{\Phi,k}$ from $S^{(0)}$ to $Y$ as
	\begin{align*}
	&\Lambda_{\Phi} (k,\rho_1,\ldots,\rho_{\sff}) 
	= \Gamma_{\Phi,k} (\rho_{[\sff]})\\
	&= \Tr_{S^{(\sr)},E}  \cD \ast \cE ( \rho_{[\sff]} \otimes \cD^{(0)} (|k\rangle\langle k|  \otimes |0\rangle\langle 0|) ),
	\end{align*}
	where $\cD \ast \cE = ( \cD^{(\sr)} \circ \cE^{(\sr)} )\circ\cdots \circ ( \cD^{(1)} \circ \cE^{(1)} ) $.
The QPIR protocol $\Phi$ should satisfy the following conditions.
\begin{itemize}[leftmargin=1.5em]
\item \textbf{Correctness}: 
	When $|\psi_k\rangle \langle \psi_k |$ denotes a purification of $\rho_k$ with the reference system $R_k$,
	the correctness is 
\begin{align}
	\Gamma_{\Phi,k}\otimes \id_{R_k}(\rho_{[\sff]\setminus \{k\}} \otimes |\psi_k\rangle \langle \psi_k |)
	 = |\psi_k\rangle \langle \psi_k |\Label{BJI}
\end{align}
	for any $K=k$ and any state $\rho_{[\sff]}$.

\item \textbf{Secrecy}: 
	When the final state on $S^{(\sr)} \otimes R_{[\sff]}$ 
	with the target index $K=k$ is denoted by 
	$\rho_{S^{(\sr)} R_{[\sff]}}^k,$
	the secrecy is 
		\begin{align}
		\rho_{S^{(\sr)} R_{[\sff]}}^k
		&= \rho_{S^{(\sr)} R_{[\sff]}}^{k'}
		\Label{eq:sec2cc}
		\end{align}
		for any $k,k'$.
\end{itemize}
The communication complexity of the one-server multi-round 
Q-QPIR is written as $\CC = \sum_{i=1}^{\sr}  \log |Q^{(i)}| +  \log |A^{(i)}|$.

\begin{theo} \Label{theo:multiround}
For any multi-round 
Q-QPIR protocol $\Phi$,
the communication complexity $\CC$ is lower bounded by $\sum_{\ell=1}^{\sff} \log |X_\ell|$, where $X_\ell$ is the system of the $\ell$-th message $\rho_\ell$.
\end{theo}

For the proof of Theorem \ref{theo:multiround},
	we prepare the following lemmas.
\begin{lemm} \Label{lemm:recurs}
$H(A^{(i)}) +  H(Q^{(i+1)})  \geq H(T^{(i+1)} ) - H(T^{(i)})$.
\end{lemm}
\begin{proof}
Lemma \ref{lemm:recurs} is shown by the relation
\begin{align*}
 &H(A^{(i)}) + H(T^{(i)}) +  H(Q^{(i+1)}) \\
 &\stackrel{\mathclap{(b)}}{\geq} H(A^{(i)} T^{(i)} ) +  H(Q^{(i+1)}) \\
 &\stackrel{\mathclap{(c)}}{=} H(Q^{(i+1)} T^{(i+1)} ) +  H(Q^{(i+1)}) \\
 &\stackrel{\mathclap{(d)}}{\ge} H(T^{(i+1)} ).
%
\end{align*}
Here, $(b)$, $(c)$, and $(d)$ express the respective properties presented in 
Proposition \ref{PP3}.
\end{proof}

\begin{lemm} \Label{lemm:prrrd}
The relation 
$H( R_{[\sff]} S^{(\sr)}) \ge \sum_{\ell=1}^\sff H( R_{\ell})$ 
holds.
\end{lemm}
\begin{proof}
Given the user's input $k$, 
Correctness \eqref{BJI} guarantees that
the final state on $R_k\otimes Y$ is a pure state, 
and therefore,
	$R_k$ is independent of any system except for $Y$.
Thus, $R_{k}$ is independent of 
$R_{[\sff]\setminus \{k\} } S^{(\sr)}$.
The secrecy condition \eqref{eq:sec2cc} guarantees that
the final state on $R_{[\sff]} \otimes S^{(\sr)}$ does not depend on $k$.
Hence, $R_1, \ldots, R_\sff$, and $ S^{(\sr)}$ are independent of each other.
Therefore, we have
\begin{align}	
H( R_{[\sff]} S^{(\sr)} ) = H(S^{(\sr)})+ \sum_{\ell=1}^\sff H( R_{\ell})  
\ge \sum_{\ell=1}^\sff H( R_{\ell})
\Label{eq:222f}.
\end{align}
\end{proof}

\begin{proof}[Proof of Theorem~\ref{theo:multiround}]
We choose the initial state on $R_\ell\otimes X_{\ell}$ to be the maximally entangled state for 
$\ell=1, \ldots, \sff$.
From Lemmas~\ref{lemm:recurs} and \ref{lemm:prrrd}, we derive the following inequalities:
\begin{align}
	&\CC \geq \sum_{i=1}^{\sr} \big(H(A^{(i)}) + H(Q^{(i)})\big)\nonumber\\
	 &= H(A^{(\sr)}) + H(Q^{(1)}) 
	 + \sum_{i=1}^{\sr-1} \big(H(A^{(i)}) + H(Q^{(i+1)}) \big)	\nonumber\\
	 &\geq H(A^{(\sr)}) + H(Q^{(1)}) + H(T^{(\sr)}) - H(T^{(1)}) \Label{eq:lemmapp}\\
	 &
	 = H(A^{(\sr)}) + H(T^{(\sr)}) \Label{eq:purrr}\\
	 &\stackrel{\mathclap{(b)}}{\geq} H(A^{(\sr)} T^{(\sr)}) \stackrel{\mathclap{(a)}}{=} H(R_{[\sff]} S^{(\sr)}) \nonumber\\
	 &\geq \sum_{\ell=1}^\sff H( R_{\ell})=\sum_{\ell=1}^{\sff} \log |X_\ell|
	 ,\Label{eq:samewa} 
\end{align}
where 
$(a)$ and $(b)$ express the respective properties presented in 
Proposition \ref{PP3}.
In addition,
\eqref{eq:lemmapp} is obtained by applying Lemma~\ref{lemm:recurs} for all $i=1,\ldots, \sr-1$.
The step \eqref{eq:purrr} follows from $H(Q^{(1)}) = H(T^{(1)})$ which 
 		holds due to the property $(a)$ in Proposition \ref{PP3}, 		
 		because the state on $Q^{(1)} T^{(1)}$ is the pure state
 		as  the state on $Q^{(0)} T^{(0)}$ is the pure state.
The step \eqref{eq:samewa} follows from Lemma~\ref{lemm:prrrd}.
\end{proof}

\section{Q-QPIR Protocol with Prior Entanglement under honest-server model} \Label{sec:ent}
In the previous section, we proved that the trivial solution is optimal 
even in the final-state criterion under the honest one-server model of Q-QPIR.
In this section, we construct a Q-QPIR protocol with lower communication complexity 
under various secrecy models than the trivial solution
when we allow shared entanglement between the user and the server.

Let $\sm = \sum_{\ell=1}^{\sff} \log |X_\ell|$ be the size of all messages.
To measure the amount of the prior entanglement,
we count sharing one copy of $|\sI_2\rrangle = (1/\sqrt{2})(|00\rangle + |11\rangle)$ as an {\em ebit}. 
Accordingly, we count sharing the state $|\sI_d\rrangle \in \mathbb{C}^d \otimes \mathbb{C}^d$ as $\log d$ ebits.

\begin{theo} \Label{theo:prot} 
Suppose there exists a C-QPIR protocol under a certain secrecy model
with communication complexity $f(d_1, \ldots, d_\sff)$
when $g(d_1, \ldots, d_\sff)$-ebit prior entanglement is shared between the user and the server.
Then, there exists a Q-QPIR protocol under the same secrecy model with communication complexity 
$f(d_1^2, \ldots, d_\sff^2)$
when $\sm+g(d_1, \ldots, d_\sff)$-ebit prior entanglement is shared between the user and the server.
\end{theo}

The protocol satisfying Theorem \ref{theo:prot} is a simple combination of quantum teleportation \cite{BBCJPW93} and any C-QPIR protocol.
For the description of the protocol, we use the generalized Pauli operators and maximally entangled state for $d$-dimensional systems defined in \eqref{eq:defs-pauli}.
Hence, the type of guaranteed secrecy in the original 
C-QPIR protocol is inherited to the converted QPIR protocol.
We construct the Q-QPIR protocol satisfying Theorem \ref{theo:prot} as follows.
\begin{prot} \Label{prot:ent}
Let $\Phi_{\mathrm{cl}}$ be a C-QPIR protocol 
	and $d_1,\ldots, d_{\sff}$ be the size of the $\sff$ classical messages. 
From this protocol, we construct a Q-QPIR protocol as follows.

Let $X_1,\ldots, X_\sff$ be the quantum systems with dimensions $d_1,\ldots, d_{\sff}$, respectively, and 
 $\rho_1,\ldots, \rho_{\sff}$ be the quantum message states on systems $X_1,\ldots, X_\sff$.
The user and the server share the maximally entangled states $|\sI_{d_\ell}\rrangle$, defined in \eqref{eq:defs-pauli}, on ${Y_\ell\otimes Y_\ell'}$ for all $\ell\in[\sff]$,
	where $Y_{[\sff]}$ and $Y_{[\sff]}'$ are possessed by the user and the server, respectively.
	
The user and the server perform the following steps.
\begin{enumerate}[leftmargin=1.5em]
\item[1)] \textbf{Preparation}:
		For all $\ell\in[\sff]$, the server performs the generalized Bell measurement
		$\mathbf{M}_{\sX\sZ,d_\ell}$, defined in \eqref{mes},
		on $X_\ell\otimes Y_\ell'$, where the measurement outcome is written as $m_\ell = (a_\ell, b_\ell)\in [0:  d_\ell-1]^2$.

\item[2)]\textbf{Use of C-QPIR protocol}: 
	The user and the server perform the 
	C-QPIR protocol $\Phi_{\mathrm{cl}}$ to retrieve $m_k = (a_k,b_k)$.
\item[3)]\textbf{Reconstruction}:
 The user recovers the $k$-th message $\rho_k$ by applying $\sX_{d_k}^{-a_k}\sZ_{d_k}^{b_k} $ on $Y_k$.
\QEDA
\end{enumerate}
\end{prot}
The correctness of the protocol is guaranteed by the correctness of the teleportation protocol and the 
C-QPIR protocol $\Phi_{\mathrm{cl}}$.
When the $\ell$-th message state is prepared as $\rho_\ell$ and its purification $|\phi_\ell\rangle$ is denoted with the reference system $R_\ell$,
	after Step~1, the states on $R_\ell\otimes Y_\ell$ is 
	\begin{align}
	( \sI\otimes  \sX_{d_\ell}^{a_\ell} \sZ_{d_\ell}^{-b_\ell}) |\phi_\ell\rangle
	\end{align}
for all $\ell\in[\sff]$.
Thus after Step~3, the targeted state $|{\phi_k}\rangle$ is recovered in $R_k\otimes Y_k$.

To analyze the secrecy of Protocol~\ref{prot:ent}, 
	note that only Step~2 has the communication between the user and the server.
Thus the secrecy of Protocol~\ref{prot:ent} is guaranteed by the secrecy of the underlying protocol $\Phi_{\mathrm{cl}}$.

Protocol \ref{PR2} (Protocol \ref{PR2B}) 
is a one-round C-QPIR protocol in the final-state criterion
under the honest-server model (the specious-server model) with input states ${\cal C}$
with communication complexity $2 \log\sff+\log d$ ($4 \log\sff+2 \log d$).
Therefore, 
the combination of Protocols \ref{PR2} and \ref{prot:ent} and 
the combination of Protocols \ref{PR2B} and \ref{prot:ent}
yield the following corollary. 
\begin{coro}\Label{CO2}
There exists a Q-QPIR protocol 
	with communication complexity $2 \log\sff +\log d^2=2\log \sff d $
	($4 \log\sff +2 \log d^2=4\log \sff d $)
	and prior entanglement $\sm $ that satisfies the secrecy in the final-state criterion
under the honest-server model (the specious-server model).
	When $d$ is a constant,
the communication complexity is $2 \log \sm+o(\sm)$ ($4 \log \sm+o(\sm)$).
\end{coro}

\begin{proof}
The case under the honest-server model is trivial.
Hence, we show the desired statement under the specious-server model.

Assume that the server makes a specious attack.
The user's state at the end of Step 2) of Protocol \ref{prot:ent}
is the pair of 
entanglement halves $\sigma_1$
and the state transmitted at Step 2) of Protocol \ref{PR2B}
$\sigma_2$.
Due to the specious condition, 
the state $\sigma_1$ needs to be 
one of the states
$\{\sX^a\sZ^b \rho_K(\sX^a\sZ^b)^\dagger\}_{(a,b) \in [0:d-1]^2}$
with equal probability.
That is, using the random variable $(a,b) \in [0:d-1]^2$ under the uniform distribution,
the state $\sigma_1$ is written as $\sX^a\sZ^b \rho_K(\sX^a\sZ^b)^\dagger$.
Hence, the state $\sigma_2$ needs to be decided according to 
the random variable $(a,b)$ in the same way as the honest case.
That is, the state $\sigma_2$ satisfies the condition for the state transmitted by
a specious server of Protocol \ref{PR2B} at Step 2). 
Since Protocol \ref{PR2B} satisfies the secrecy under the final-state criterion
under the specious-server model with input states ${\cal C}$,
the specious server obtains no information in the final state. 
That is, the combined Q-QPIR protocol with prior entanglement
satisfies the secrecy under the final-state criterion
under the specious-server model.
\end{proof}

Combining Theorem \ref{theo:prot} and Corollary \ref{CCor1},
we obtain the following corollary.

\begin{coro}\Label{CO3}
There exists a Q-QPIR protocol 
	with communication complexity $O( \log \sm)$
	and prior entanglement of $\Theta(\sm)$ ebits 
	that satisfies the secrecy in the all-round criterion
under the honest-server model
	when the message size $d$ is fixed to a constant.
\end{coro}

One property of Protocol~\ref{prot:ent} is that all other states in the server are destroyed at Step~1.
This is a disadvantage for the server but an advantage for the user since the user can retrieve other states $\rho_\ell$ 
if the user could retrieve classical information $m_\ell \in [0:d_{\ell}-1]^2$
corresponding to the state $\rho_\ell$. 

\section{Conclusion} \Label{sec:conclusion}
We have shown an exponential gap for the communication complexity of one-server Q-QPIR in the final-state criterion or
under the honest-server model between the existence and the non-existence of prior entanglement.
For this aim, as the first step,
we have proposed an efficient one-server one-round C-QPIR protocol
in the final-state criterion.
Also, we have shown that the protocols proposed in \cite{KLLGR16}
satisfies the secrecy in the all-round criterion 
under the honest server model.
Then, as the second step,
we have proved that the trivial solution of downloading all messages is optimal even in the final-state criterion 
for honest one-server Q-QPIR,
	which is a similar result to that of classical PIR 
	but different from C-QPIR.
As the third step,
we have developed a conversion from any C-QPIR protocol 
to a Q-QPIR protocol, 
which yields an efficient Q-QPIR protocol 
with prior entanglement from a C-QPIR protocol.
The proposed protocols exhibit an exponential improvement over the Q-QPIR's trivial solution.

In fact, Protocols \ref{PR2} and \ref{PR2B} work as 
one-server one-round C-QPIR protocol
in the final-state criterion 
under the honest-server model or the specious-server model.
However, Theorem \ref{theo:multiround} shows that
no analogy of Protocol \ref{PR2} nor \ref{PR2B} works for Q-QPIR protocol
under similar settings without prior entanglement.
This impossibility is caused by the non-cloning property of the quantum system,
i.e., the property that 
the noiseless channel has no information leakage to the third party,
because the proof of Theorem \ref{theo:multiround} relies on the fact that
noiseless quantum communication ensures that 
the entropy of the final state on the third party 
is equal to the entropy of 
the final state on the composite system comprising the output system and the reference system. 
This impossibility is one of the reasons for our obtained exponential gap.

The above exponential gap has been established under three problem settings.
The first and the second are 
the final-state criterion under the honest-server model
and under the specious-server model.
The third is the all-round criterion under the honest-server model.
In other words, 
other problem settings do not have such an exponential improvement by using prior entanglement.
This exponential improvement is much larger than the improvement 
achieved through the use of dense coding \cite{BW}.
This exponential improvement can be considered as a useful
application of use of prior entanglement.
It is an interesting open problem 
to find similar exponential improvement by using
 prior entanglement.

\section*{Acknowledgement}
SS was supported by JSPS Grant-in-Aid for JSPS Fellows No.\ JP20J11484.
FLG was partially supported by JSPS KAKENHI grants Nos. JP20H04139 and JP21H04879.
MH was supported in part by the National
Natural Science Foundation of China (Grants No. 62171212) and
Guangdong Provincial Key Laboratory (Grant No. 2019B121203002).


\begin{thebibliography}{99}
\bibitem{BBCJPW93}
C. H. Bennett, G. Brassard, C. Cr\'epeau, R. Jozsa, A. Peres, and W. K. Wootters, 
``Teleporting an unknown quantum state via dual classical and Einstein-Podolsky-Rosen channels,'' 
{\em Phys. Rev. Lett.}, {\bf 70}(13), 1895--1899, 1993.

\bibitem{BW}
C. H. Bennett and S. J. Wiesner, 
``Communication via One- and Two-Particle Operators on Einstein-
Podolsky-Rosen States,'' 
{\em Phys. Rev. Lett.}, {\bf 69}, 2881 (1992).

\bibitem{WDLLL}
C. Wang, F.-G. Deng, Y.-S. Li, X.-S. Liu, and G.-L. Long, 
``Quantum secure direct communication with high-dimension quantum superdense coding,'' 
{\em Phys. Rev. A}, {\bf 71}, 044305 (2005).

\bibitem{WLH}
J. Wu, G.-L. Long, and M. Hayashi, 
``Quantum secure direct communication with private dense coding using a general preshared quantum state,'' 
{\em Physical Review Applied}, {\bf 17}, 064011 (2022).

\bibitem{KdW03}
I. Kerenidis and R. de Wolf,
``Exponential lower bound for 2-query locally decodable codes via a quantum argument,'' 
{\em Proc. 35th ACM symposium on Theory of computing (STOC' 03)}, 
June 2003. 
pp 106 -- 115.

\bibitem{KdW04}
I. Kerenidis and R. de Wolf, ``Quantum symmetrically-private information retrieval,''
{\em Information Processing Letters}, {\bf 90}, 109--114, 2004.

\bibitem{Ole11}
L. Olejnik, 
``Secure quantum private information retrieval using phase-encoded queries,'' 
{\em Phys. Rev. A}, {\bf 84}, 022313, 2011.

\bibitem{BB15}
\"{A}. Baumeler and A. Broadbent, 
``Quantum Private Information Retrieval has linear communication complexity,'' 
{\em Journal of Cryptology}, {\bf 28}, 161--175, 2015.

\bibitem{LeG12}
F. Le Gall, 
``Quantum Private Information Retrieval with Sublinear Communication Complexity,'' 
{\em Theory of Computing}, {\bf 8}(16), 369 -- 374, 2012.

\bibitem{KLLGR16}
I. Kerenidis, M. Lauri\`{e}re, F. Le Gall, and M. Rennela, 
``Information cost of quantum communication protocols,'' 
{\em Quantum information \& computation,} 
{\bf 16}(3-4), 181--196, 2016.

\bibitem{ABCGLS19}
D. Aharonov, Z. Brakerski, K.-M. Chung, A. Green, C.-Y. Lai, and O. Sattath,
``On Quantum Advantage in Information Theoretic Single-Server PIR,'' 
{\em In: Ishai Y., Rijmen V. (eds) EUROCRYPT 2019,} Springer, Cham, vol.~11478, 2019.

\bibitem{SH19}
S. Song and M. Hayashi, 
``Capacity of Quantum Private Information Retrieval with Multiple Servers,'' 
{\em IEEE Trans. Inf. Theory}, 
{\bf 67}, no. 1,  452 -- 463, 2021.
                
\bibitem{SH19-2}
S. Song and M. Hayashi, 
``Capacity of Quantum Private Information Retrieval with Collusion of All But One of Servers,''
{\em IEEE Journal on Selected Areas in Information Theory},  
{\bf 2}, no. 1, 380 -- 390, 2021.
                
\bibitem{SH20}
S. Song and M. Hayashi, 
``Capacity of Quantum Private Information Retrieval with Colluding Servers,''
{\em IEEE Trans. Inf. Theory}, 
{\bf 67}, no. 8,  5491 -- 5508, 2021.

\bibitem{AHPH20}
M. Allaix, L. Holzbaur, T. Pllaha, and C. Hollanti,
``Quantum Private Information Retrieval From Coded and Colluding Servers,''
{\em IEEE Journal on Selected Areas in Information Theory,} 
{\bf 1}, no.~2, 599 -- 610, 2020.
 
\bibitem{ASHPHH21}
 M. Allaix, S. Song, L. Holzbaur, T. Pllaha, M. Hayashi, and C. Hollanti,
``On the Capacity of Quantum Private Information Retrieval from MDS-Coded and Colluding Servers,''
{\em IEEE Journal on Selected Areas in Communications},
{\bf 40}, no. 3, 885 -- 898, 2022.

\bibitem{KL20}
W. Y. Kon and  C. C. W. Lim,
``Provably Secure Symmetric Private Information Retrieval with Quantum Cryptography,''
{\em Entropy}, {\bf 23}, no.~1, 54, 2021.

\bibitem{WKNL21-1}
 C. Wang, W. Y. Kon, H. J. Ng, and C. C. Lim, 
 ``Experimental symmetric private information retrieval with measurement-device-independent quantum network'',
{\em Light. Sci. Appl.}, {\bf 11}, 268 (2022).
 
\bibitem{WKNL21-2}
C. Wang, W. Y. Kon, H. J. Ng, and C. C. Lim, 
``Experimental symmetric private information retrieval with quantum key distribution,'' 
{\em Quantum Information and Measurement VI 2021, 
F. Sciarrino, N. Treps, M. Giustina, and C. Silberhorn, eds., Technical Digest Series}, 
Optica Publishing Group, 2021.

\bibitem{Wie83}
S. Wiesner, ``Conjugate Coding,'' {\em SIGACT News}, {\bf 15}(1):78 -- 88, 1983.

\bibitem{GC01}
D. Gottesman and I. Chuang, 
``Quantum Digital Signatures,'' 2001, 
arXiv:quant-ph/0105032

\bibitem{Moc07}
C. Mochon, ``Quantum weak coin flipping with arbitrarily small bias,'' 2007, 
arXiv:0711.4114.

\bibitem{CK09}
A. Chailloux and I. Kerenidis. 
``Optimal Quantum Strong Coin Flipping.'' 
{\em Proc. 50th Annual IEEE Symposium on Foundations of Computer Science, FOCS 2009}, Atlanta, Georgia, USA, October 25-27, 2009.
pp. 527 -- 533. 

\bibitem{ACG+16}
D. Aharonov, A. Chailloux, M. Ganz, I. Kerenidis, and L. Magnin. 
``A Simpler Proof of the Existence of Quantum Weak Coin Flipping with Arbitrarily Small Bias,'' {\em SIAM J. Comput.}, {\bf 45}(3):
633 -- 679, 2016.

\bibitem{CMS99}
C. Cachin, S. Micali, and M. Stadler, 
``Computationally Private Information Retrieval with Polylogarithmic Communication,'' 
{\em Advances in Cryptology - EUROCRYPT '99,} pp. 402--414, 1999. 

\bibitem{Lipmaa10}
H. Lipmaa, ``First CPIR Protocol with Data-Dependent Computation,'' 
{\em Proceedings of the 12th International Conference on Information Security and Cryptology,} 
pp. 193--210, 2009. 

\bibitem{BS03}
A. Beimel and Y. Stahl, ``Robust information-theoretic private information retrieval,'' 
{\em Proceedings of the 3rd International Conference on Security in Communication Networks (SCN'02)}, pp. 326--341, 2003. 

\bibitem{Yekanin07}
S. Yekhanin, ``Towards 3-query locally decodable codes of subexponential length,'' 
{\em Journal of the ACM}, {\bf 55}, no.  1, 1-- 6, 2008.

\bibitem{DGH12}
C. Devet, I. Goldberg, and N. Heninger, ``Optimally Robust Private Information Retrieval,'' 
{\em 21st USENIX Security Symposium}, August 2012.

\bibitem{CHY15}
T. H. Chan, S.-W. Ho, and H. Yamamoto, 
``Private information retrieval for coded storage,''
{\em Proc. IEEE International Symposium on Information Theory (ISIT2015 )}, 
Hong Kong, China, June, 14--19, 2015.
pp. 2842-2846,

\bibitem{SJ17}
H. Sun and S. Jafar, 
``The capacity of private information retrieval,'' 
{\em IEEE Trans. Inf. Theory}, 
{\bf 63}, no. 7, 4075 -- 4088, 2017.

\bibitem{SJ17-2}
H. Sun and S. Jafar, 
``The Capacity of Symmetric Private Information Retrieval,'' 
{\em 2016 IEEE Globecom Workshops (GC Wkshps)}, Washington, DC, 2016, pp. 1--5.

\bibitem{SJ18}
H. Sun and S. Jafar, 
``The capacity of robust private information retrieval with colluding databases,''
{\em IEEE Trans. Inf. Theory}, 
vol. 64, no. 4, 2361 -- 2370, 2018.

\bibitem{FHGHK17}
R. Freij-Hollanti, O. W. Gnilke, C. Hollanti, and D. A. Karpuk, 
``Private information retrieval from coded databases with colluding servers,'' 
{\em SIAM J. Appl. Algebra Geometry}, vol. 1, no. 1, pp. 647--664, 2017.

\bibitem{KLRG17}
S. Kumar, H.-Y. Lin, E. Rosnes, and A. Graell i Amat, 
``Achieving maximum distance separable private information retrieval capacity with linear codes,'' {\em IEEE Trans. Inf. Theory}, 
{\bf 65}, no.~7,  4243 -- 4273, 2019.

\bibitem{WS17}
Q. Wang and M. Skoglund, 
``Symmetric private information retrieval for MDS coded distributed storage,'' 
{\em Proceedings of 2017 IEEE International Conference on Communications (ICC)}, 
pp. 1--6, May 2017.

\bibitem{LKRG18}
H.-Y. Lin, S. Kumar, E. Rosnes, and A. Graell i Amat, 
``An MDS-PIR capacity-achieving protocol for distributed storage using non-MDS linear codes,'' 
{\em Proc. IEEE International Symposium on Information Theory (ISIT2018)}, 
Talisa Hotel in Vail, Colorado, USA, June, 17 -- 22, 2018.
pp. 966 -- 970, 
               
\bibitem{BU18}
K. Banawan and S. Ulukus,
``The Capacity of Private Information Retrieval from Coded Databases,''
{\em IEEE Trans. Inf. Theory}, 
{\bf 64}, no. 3, 1945 -- 1956, 2018,

\bibitem{TSC18-2}
C. Tian, H. Sun and J. Chen, 
``A Shannon-Theoretic Approach to the Storage-Retrieval Tradeoff in PIR Systems,'' 
{\em Proc. IEEE International Symposium on Information Theory (ISIT2018)}, 
 Talisa Hotel in Vail, Colorado, USA, June, 17 -- 22, 2018.
pp. 1904--1908.

\bibitem{Tandon17}
R. Tandon, 
``The capacity of cache aided private information retrieval,'' 
{\em Proc. 2017 55th Annual Allerton Conference on Communication, Control, and Computing (Allerton)}, 
pp. 1078--1082, 2017.

\bibitem{BU19}
K. Banawan and S. Ulukus, 
``The capacity of private information retrieval from byzantine and colluding databases,'' 
{\em IEEE Trans. Inf. Theory}, 
{\bf 65}, no. 2, 1206--1219, 2019.

\bibitem{L19}
L. Holzbaur, R. Freij-Hollanti, J. Li, and C. Hollanti, 
``Towards the Capacity of Private Information Retrieval from Coded and Colluding Servers,'' 
{\em IEEE Trans. Inf. Theory}, 
{\bf 68}, no. 1, 517 -- 537, 2022

\bibitem{KGHERS19}
S. Kadhe, B. Garcia, A. Heidarzadeh, S. El Rouayheb, A. Sprintson,
``Private information retrieval with side information,'' 
{\em IEEE Trans. Inf. Theory}, 
{\bf 66} no. 4, 2032--2043, 2019.

\bibitem{TGKFH19}
R. Tajeddine, O. W. Gnilke, D. Karpuk, R. Freij-Hollanti and C. Hollanti, 
``Private Information Retrieval From Coded Storage Systems With Colluding, Byzantine, and Unresponsive Servers," 
{\em IEEE Trans. Inf. Theory}, {\bf 65}, no. 6, 3898 -- 3906, 2019. 

\bibitem{GLM}
V. Giovannetti, S. Lloyd, and L. Maccone,
``Quantum private queries,''
{\em Phys. Rev. Lett.}, {\bf 100}, 230502 (2008)

\bibitem{DNS10}
F. Dupuis, J.B. Nielsen, and L. Salvail, 
``Secure two-party quantum evaluation of unitaries against specious adversaries'', 
in {\em Proc. 30th Annual Conference on Advances in Cryptology (CRYPTO'10)},
(Springer, Berlin, 2010), pp. 685--706, 2010.

\bibitem{CGKS98} 
B. Chor, O. Goldreich, E. Kushilevitz, and M. Sudan, 
``Private information retrieval,'' 
{\em Journal of the ACM,} {\bf 45}(6), 965--981, 1998.

\bibitem{Hayashi}
M. Hayashi,
{\em Quantum Information Theory: Mathematical Foundation}, 
{\em Graduate Texts in Physics}, Springer (2017).

\bibitem{KI}
 M. Koashi and N. Imoto, 
 {\em Phys. Rev. A}, {\bf 66}, 022318 (2002).

\bibitem{HJPW}
 P. Hayden, R. Jozsa, D. Petz, and A. Winter, 
{\em Commun. Math. Phys.}, {\bf 246}, 359 (2004).

\bibitem{BMPV}
R. Blume-Kohout, H. K. Ng, D. Poulin, and L. Viola, 
{\em Phys. Rev. A}, {\bf 82}, 062306 (2010).

\bibitem{WSM}
E. Wakakuwa, A. Soeda, and M. Murao, 
{\em IEEE Trans. Inf. Theory}, 
{\bf 63}, 1280 (2017).

\bibitem{269}
M. Horodecki, 
``Limits for compression of quantum information carried by ensembles of mixed states,'' 
{\em Phys. Rev. A,}  {\bf 57}, 3364--3369, 1998.

\bibitem{35}
H. Barnum, C. M. Caves, C. A. Fuchs, R. Jozsa, and B. Schumacher, 
``On quantum coding for ensembles of mixed states,'' 
{\em J. Phys. A Math. Gen.,} {\bf 34}, 6767--6785, 2001.

\bibitem{201}
M. Hayashi, 
``Exponents of quantum fixed-length pure state source coding,'' 
{\em Phys. Rev. A,} {\bf 66}, 032321, 2002.
  
\bibitem{288}
R. Jozsa and B. Schumacher, 
``A new proof of the quantum noiseless coding theorem,'' 
{\em J. Mod. Opt.,} {\bf 41}(12), 2343--2349, 1994.
	
\bibitem{423}
B. Schumacher, 
``Quantum coding,'' {\em Phys. Rev. A,} {\bf 51}, 2738--2747, 1995.
	
\bibitem{290}
R. Jozsa, M. Horodecki, P. Horodecki, and R. Horodecki, 
``Universal quantum information compression,'' 
{\em Phys. Rev. Lett.,} {\bf 81}, 1714, 1998.

\bibitem{AYH}
H. Arai, Y. Yoshida, and M. Hayashi,
``Perfect Discrimination of Non-Orthogonal Separable Pure States on Bipartite System in General Probabilistic Theory,''
{\em Journal of Physics A: Mathematical and Theoretical}, 
{\bf 52}, 465304 (2019).

\bibitem{BDLT}
H. Barnum, O. C. Dahlsten, M. Leifer, and B. Toner,
``Nonclassicality without entanglement enables bit commitment",
{\em IEEE Information Theory Workshop, IEEE}, 2008, pp. 386–390.

\bibitem{Kilian}
J. Kilian, 
``Founding crytpography on oblivious transfer''. 
{\em Proc. 20th ACM symposium on Theory of computing (STOC' 88)}, 
June 1988. 
pp 20 -- 31.

\bibitem{Salvail}
L. Salvail,  
``The Search for the Holy Grail in Quantum Cryptography,'' 
In Lectures on Data Security: 
Modern Cryptology in Theory and Practice; 
Damg\r{a}rd, I.B., Ed.; Springer Berlin Heidelberg: Berlin, Heidelberg, 1999; pp. 183–216. 

\end{thebibliography}
\end{document}